\documentclass[11pt]{article}
\usepackage[margin=0.9in]{geometry}

\usepackage{varioref}
\usepackage[usenames,svgnames,xcdraw,table]{xcolor}
\definecolor{DarkBlue}{rgb}{0.1,0.1,0.5}
\definecolor{DarkGreen}{rgb}{0.1,0.5,0.1}

\usepackage{hyperref}
\hypersetup{
   colorlinks   = true,
 linkcolor    = DarkBlue, 
 urlcolor     = DarkBlue, 
	 citecolor    = DarkGreen 
}

\usepackage{appendix}

\usepackage{mathtools}

\usepackage{algorithmic}

\usepackage{algorithm}
\usepackage{array}

\usepackage[font={small,it}]{caption}

\usepackage{graphicx,epsfig,amsmath,latexsym,amssymb,verbatim}

\usepackage{nicefrac}

\newcommand{\extra}[1]{}

\usepackage{amsmath,amssymb,amsfonts}
\usepackage{amsthm}
\usepackage{thmtools,thm-restate}
\usepackage{enumitem}

\newtheorem{corollary}{Corollary}

\newtheorem{definition}{Definition}
\newtheorem{lemma}{Lemma}

\newtheorem{proposition}{Proposition}
\newtheorem{claim}{Claim}
\newtheorem*{remark2}{Remark}

\def\squareforqed{\hbox{\rlap{$\sqcap$}$\sqcup$}}
\def\qed{\ifmmode\squareforqed\else{\unskip\nobreak\hfil
\penalty50\hskip1em\null\nobreak\hfil\squareforqed
\parfillskip=0pt\finalhyphendemerits=0\endgraf}\fi}
\def\endenv{\ifmmode\;\else{\unskip\nobreak\hfil
\penalty50\hskip1em\null\nobreak\hfil\;
\parfillskip=0pt\finalhyphendemerits=0\endgraf}\fi}

\renewenvironment{proof}{\noindent \textbf{{Proof~} }}{\qed\medskip}
\newenvironment{proof+}[1]{\noindent \textbf{{Proof #1~} }}{\qed\medskip}
\newenvironment{remark}{\noindent \textit{{Remark.~}}}{\qed}

\mathchardef\ordinarycolon\mathcode`\:
\mathcode`\:=\string"8000
\def\vcentcolon{\mathrel{\mathop\ordinarycolon}}
\begingroup \catcode`\:=\active
  \lowercase{\endgroup
  \let :\vcentcolon
  }

\newcommand{\nc}{\newcommand}

\DeclareMathOperator*{\argmin}{arg\,min}

\nc{\barA}{\overline{A}}
\nc{\barB}{\overline{B}}
\nc{\barC}{\overline{C}}
\nc{\barD}{\overline{D}}
\nc{\barR}{\overline{R}}
\nc{\barX}{\overline{X}}
\nc{\barY}{\overline{Y}}
\nc{\barU}{\overline{U}}

\usepackage{bbm}

\usepackage{verbatim}

\usepackage[nottoc]{tocbibind}

\usepackage{thm-restate}
\usepackage{bigints}

\newcommand{\NSW}{\mathrm{NSW}}
\newcommand{\SW}{\mathrm{SW}}
\newcommand{\RW}{\mathrm{EW}}

\newcommand{\I}{\mathcal{I}}

\newcommand{\Cake}{\mathcal{C}}

\newcommand{\Eval}{\mathrm{Eval}}
\newcommand{\Cut}{\mathrm{Cut}}
\newcommand{\RD}{\mathrm{RD}}
\newcommand{\MK}{\mathrm{MK}}

\begin{document}

\title{{\bfseries Fair Cake Division Under Monotone Likelihood Ratios}}
\author{Siddharth Barman\thanks{Indian Institute of Science. {\tt barman@iisc.ac.in}} \and Nidhi Rathi\thanks{Indian Institute of Science. {\tt nidhirathi@iisc.ac.in}}}

\date{}
\maketitle

\thispagestyle{empty}
\begin{abstract}
This work develops algorithmic results for the classic cake-cutting problem in which a divisible, heterogeneous resource (modeled as a cake) needs to be partitioned among agents with distinct preferences. We focus on a standard formulation of cake cutting wherein each agent must receive a contiguous piece of the cake. While multiple hardness results exist in this setup for finding fair/efficient cake divisions, we show that, if the value densities of the agents satisfy the \emph{monotone likelihood ratio property} (MLRP), then strong algorithmic results hold for various notions of fairness and economic efficiency. 

Addressing cake-cutting instances with MLRP, first we develop an algorithm that finds cake divisions (with connected pieces) that are envy-free, up to an arbitrary precision. The time complexity of our algorithm is polynomial in the number of agents and the bit complexity of an underlying Lipschitz constant. We obtain similar positive results for maximizing social (utilitarian) and egalitarian welfare. In addition, we show that, under MLRP, the problem of maximizing Nash social welfare admits a fully polynomial-time approximation scheme (FPTAS).

Many distribution families bear MLRP. In particular, this property holds if all the value densities belong to any one of the following  families: Gaussian (with the same variance), linear, binomial, Poisson, and exponential distributions. Furthermore, it is known that linear translations of any log-concave function satisfy MLRP. Therefore, our results also hold when the value densities of the agents are linear translations of the following (log-concave) distributions: Laplace, gamma, beta, Subbotin, chi-square, Dirichlet, and logistic. Hence, through MLRP, the current work obtains novel cake-cutting algorithms for multiple distribution families.
\end{abstract}

\clearpage

\thispagestyle{empty}
\tableofcontents
\clearpage

\setcounter{page}{1}

\section{Introduction}

Cake division is a quintessential model in the study of fair division. This setup captures the allocation of a divisible resource (metaphorically, the cake) among agents with equal entitlements, but distinct preferences. Over the past several decades, a significant body of work in mathematics, economics, and computer science has been devoted to cake cutting; see \cite{brams1996fair, robertson1998cake, procaccia2015cake} for excellent expositions and motivating applications (e.g., border negotiations and divorce settlements) of this framework. 

Some of the central solution concepts and axiomatic characterizations in the fair-division literature stem from the cake-cutting context \cite{moulin2004fair}. Indeed, the work of Steinhaus, Banach, and Knaster \cite{S48problem}---which lays the mathematical foundations of fair division---addresses cake division. The notion of \emph{envy-freeness} was also mathematically formalized in this setup \cite{stern1958puzzle, foley1967resource}. This well-studied notion deems a cake division to be fair if every agent prefers the piece assigned to her over that of any other agent, i.e., if no agent is envious of others. 

Formally, the cake is modeled as the unit interval $[0,1]$ and the cardinal preferences of the agents over pieces of this divisible resource are specified via valuation functions: $v_i(I) \in \mathbb{R}_+$ denotes the value that an agent $i$ has for interval (piece) $I \subset [0,1]$. These valuations $v_i$s are typically assumed to be induced by value-density functions $f_i$s, i.e., $v_i(I) \coloneqq \int_{x \in I} f_i(x) dx$, for each agent $i$ and interval $I$. 

This work focuses on a standard formulation of cake division in which every agent must receive a contiguous piece of the cake. That is, the goal is to partition the cake $[0,1]$ into exactly $n$ disjoint intervals (connected pieces) and assign them  among the $n$ participating agents. This connectivity requirement is naturally motivated by settings in which a contiguous part of the resource needs to be allocated to every agent \cite{brams1996fair}; consider, e.g., division of land, transmission spectrum, or processing time on a machine. Note that a partition of the cake $[0,1]$ into intervals $I_1, I_2, \ldots, I_n$---wherein interval $I_i$ is assigned to agent $i \in [n]$---is said to be envy-free iff $v_i(I_i) \geq v_i(I_j)$ (i.e., iff $\int_{I_i} f_i \geq \int_{I_j} f_i$) for all agents $i$ and $j$. 

The appeal of envy-freeness is substantiated by strong existential results: under mild assumptions, a contiguous envy-free cake division always exists~\cite{stromquist1980cut, simmons1980private, edward1999rental}.  While these results are built upon interesting mathematical connections,\footnote{For instance, the proof by Su \cite{edward1999rental} invokes Sperner's lemma.} they are, however, nonconstructive. In fact, Stromquist \cite{stromquist2008envy} has shown that there does not exist a finite-time algorithm for finding envy-free cake divisions with connected pieces; this result holds in a setup wherein the valuations are provided through an (adversarial) oracle. In addition, the work of Deng et al. \cite{deng2012algorithmic} establishes {\rm PPAD}-hardness of finding envy-free cake divisions with contiguous pieces, under ordinal valuations. 

Algorithms for envy-free cake division remain elusive even if we relinquish the contiguity requirement. It was not until the work of Brams and Taylor~\cite{brams1995envy} that a bounded-time algorithm was obtained for noncontiguous envy-free cake division. In general, the best-known result for this problem is by Aziz and Mackenzie \cite{aziz2016discrete}, who develop a {hyper-exponential} time algorithm for finding envy-free divisions with noncontiguous pieces.\footnote{The problem of finding an \emph{approximate} envy-free division (not necessarily with connected pieces) admits a fully-polynomial time approximation scheme~\cite{lipton2004approximately}.}
 
In light of these algorithmic barriers, identification of computationally-tractable instances in the cake-cutting context stands as a meaningful direction of work. The current paper addresses this consideration and, in particular, identifies an encompassing property---called the \emph{monotone likelihood ratio property}---which enables the development of efficient algorithms for fair cake-cutting (with connected pieces).

The ordered value-density functions $(f_i, f_j)$ are said to satisfy the {monotone likelihood ratio property} (MLRP) iff, for every $x \leq y$ in the domain, we have $\nicefrac{f_j(x)}{f_i(x)}  \leq \nicefrac{f_j(y)}{f_i(y)}$. In other words, the {likelihood ratio} $\nicefrac{f_j(x)}{f_i(x)}$ is nondecreasing in the argument $x \in \mathbb{R}$. Intuitively, this property asserts that, in comparison to $f_i$, the density $f_j$ is higher  towards the right end of the domain. We note that MLRP does not require $f_i$ and $f_j$ to be monotonic (or unimodal) by themselves.

In the cake-division context, we will say that an ordered collection $(f_i)_{i \in [n]}$ of value densities (of the $n$ agents) satisfies the {monotone likelihood ratio property} iff for each $i \in [n-1]$, the likelihood ratio $\nicefrac{f_{i+1}(x)}{f_{i} (x)}$ is nondecreasing in $x \in [0,1]$. That is, the agents are indexed with the property that consecutive likelihood ratios bear MLRP. This property is transitive and, hence, in cake-division instances with MLRP, value densities $f_i$ and $f_j$ satisfy MLRP for all $i < j$. 



Many distribution families are also known to bear MLRP \cite{larsen2001introduction, casella2002statistical}. In particular, this property holds if all the value densities belong to any one of the following   families: Gaussian distributions (with the same variance but different means), Poisson distributions, binomial distributions, and single-parameter exponentials; see Appendix~\ref{appendix:mlrp-use-cases} for details. Furthermore, it is known that linear translations of any log-concave function satisfy MLRP \cite{saumard2014log}. In particular, linear translations of the following (log-concave) distributions also satisfy this property: Laplace, uniform, multivariate Gaussian, gamma, beta, Subbotin, chi-square, Dirichlet, and logistic. Hence, the current work obtains novel results for many distribution families in a unified manner. 

MLRP is a common assumption on agents' utilities and type distributions in many economic contexts; see~\cite{jewitt1991applications} for a survey. As a stylized application of MLRP in cake division, consider a setting wherein each agent $i$ has a most preferred point $\mu_i$ on the divisible resource (cake) and $i$'s valuation density decreases as a Gaussian function (with a variance parameter that is common across the agents) of the distance from $\mu_i$. Indeed, the distance here can be geographical (as in case of land division), temporal (i.e., wait time), or it can be an abstract metric. 

Considering similar single-peaked preferences, but with a linear drop in value densities, Wang and Wu~\cite{wang2019cake} developed an efficient algorithm for noncontiguous cake division. Note that while linear densities bear MLRP, this property does not hold for piecewise linear densities. Hence our results do not directly address the setting considered in \cite{wang2019cake}. However, in absence of the contiguity requirement (as is the case in \cite{wang2019cake}) one can find a fair cake division by first partitioning the cake into intervals, in each of which the agents' value densities are linear, and then applying the MLRP result separately.\footnote{Recall that, in contrast to such a result, our focus is on finding cake divisions in which each agent receives a contiguous piece of the cake.} 


We focus on cake-division instances with MLRP and develop algorithmic results for almost all the standard notions of fairness and economic efficiency. Our algorithms only require oracle access to the valuations. In particular, the developed algorithms operate under the standard Robertson-Webb model~\cite{robertson1998cake}, wherein we have access to the agents' valuations through \emph{eval} and \emph{cut} queries; see Section~\ref{section:notations} for details. MLRP implies that these cut and eval queries (functions) are $\lambda$-Lipschitz (Appendix~\ref{appendix:mlrp-lipschitz}). 

The time complexities of our algorithms depend polynomially on the the bit complexity of this Lipschitz constant $\lambda \geq 1$. Such a runtime dependency on $\log \lambda$ is unavoidable (Appendix~\ref{appendix:example-precision-loss}): there exist cake-division instances (with $\lambda$-Lipschitz cut and eval queries) wherein for all the agents the value of the cake is almost entirely concentrated in an interval $L$ of length ${1}/{\lambda}$. Here, an envy-free cake division can be obtained only by finely partitioning $L$ among the agents. In particular, the cut points that induce an envy-free allocation (and, hence, correspond to the output of a fair-division algorithm) must be $1/\lambda$ close to each other, i.e., the bit complexity of the output has to be ${\Omega} \left( \log \lambda \right)$. In fact, one can construct instances in which a contiguous envy-free division can be obtained only by cutting the cake at irrational points (Appendix~\ref{appendix:example-precision-loss}). Hence, in general, (and even under MLRP) one cannot expect an efficient algorithm that outputs an \emph{exact} envy-free division, with contiguous  pieces.\footnote{Indeed, the bit complexity of a computationally-bounded algorithm is bounded as well.} Therefore, when considering efficient algorithms for cake division, a precision loss in the output is inevitable. However, our algorithms ensure that this precision loss in value, $\eta$, is arbitrarily small; specifically, the developed algorithms run in time $\mathcal{O}\left( {\rm poly} \left(n, \log \lambda, \log \frac{1}{\eta} \right) \right)$ and, hence, the precision parameter $\eta$ can be driven exponentially close to zero in polynomial (in the bit complexity of $\eta$) time.  Note that this bit-precision issue is akin to the one faced in the convex-optimization problems (where again the optimal solutions can be irrational) and our runtime bound, with respect to the precision parameter $\eta$, is analogous to the one obtained by the ellipsoid method~\cite{grotschel2012geometric}. \\

\noindent
{\bf Our Results and Techniques:} Next, we summarize our results for various notions of fairness and (economic) efficiency.  \\

\noindent
\emph{Envy-Freeness:} We prove that, given a cake-division instance (in the Robertson-Webb query model) with MLRP, 
an envy-free allocation can be computed, up to an arbitrary precision, in time that is polynomial in the number of agents and $\log \lambda$; here $\lambda$ is the Lipschitz constant of the Robertson-Webb (cut and eval) queries. 

To establish this result, we define a class of divisions, referred to as \emph{ripple divisions} (Definition~\ref{defn:RD}), and prove that, under MLRP, every ripple division induces a contiguous envy-free cake division (Theorem~\ref{theorem:RD-EF}). Specifically, a collection of points $x_0 = 0 \leq x_1 \leq x_2 \leq x_{n-1} \leq x_n = 1$ (in the cake $[0,1]$) is said to form a {ripple division} of the cake if, for each $i \in [n-1]$, agent $i$ is indifferent between the consecutive intervals $[x_{i-1}, x_i]$ and $[x_i, x_{i+1}]$, i.e., $v_i(x_{i-1}, x_i) = v_i(x_i, x_{i+1})$. Note that a ripple division induces a contiguous cake division---by assigning interval $[x_{i-1}, x_i]$ to agent $i$---with the property that agent $i$ does not envy agent $i+1$. That is, in and of itself, a ripple division mandates absence of envy only between consecutive agents, and not between all pairs of agents. We will show that, interestingly, under MLRP this relaxation suffices--the cake division induced by a ripple division is  guaranteed to be envy free (Theorem~\ref{theorem:RD-EF}). 
Recall that the agents are indexed following the MLRP order: for each $i \in [n-1]$, the likelihood ratio $f_{i+1}/f_i$ is nondecreasing. Hence, allocating $[x_{i-1}, x_i]$ to agent $i \in [n]$ ensures that the intervals are assigned (left to right on the cake) in accordance with the MLRP order.

We establish the universal existence of ripple divisions through the intermediate value theorem, i.e., a one-dimensional fixed-point argument (Lemma~\ref{RDexistence}). Since one can use binary search to find fixed points in the one-dimensional setting, this proof in fact leads to an algorithm for finding ripple divisions and, hence, envy-free divisions. 
Indeed, the notion of ripple divisions and their connection with envy-freeness, under MLRP,  are two key contributions of this work.  \\

\noindent
\emph{Pareto Optimality:} We show that in cake-division instances with MLRP,  Pareto optimal cake divisions, with connected pieces, conform to the MLRP order (Lemma~\ref{theorem:POorder}). This structural result implies that for maximizing welfare we can restrict attention to allocations wherein the intervals are assigned (left to right on the cake) in accordance with the MLRP order. Intuitively, this leads us to a welfare-maximizing algorithm---specifically, a dynamic program---that recursively finds optimal allocations for intervals placed at the left end of the cake (i.e., for intervals of the form $[0,x]$).

We also establish an extension of Weller's theorem in the MLRP context. Weller's theorem~\cite{weller1985fair} asserts that there always exists some cake division---though, not necessarily with connected pieces---which is both envy-free (fair) and Pareto optimal. While this theorem holds in general,\footnote{Weller's theorem applies even in the absence of MLRP.} it does not guarantee that envy-freeness and Pareto optimality can be achieved together through contiguous cake divisions. We show that, by contrast, under MLRP \emph{every} contiguous envy-free division is Pareto optimal  (Theorem~\ref{theorem:ef-po}). Therefore, given a cake-division instance with MLRP, the allocation computed by our algorithm is not only envy-free but also Pareto optimal, up to an arbitrary precision. \\

\noindent
\emph{Social Welfare:} Social (utilitarian/Benthamite) welfare is a standard measure of collective value. For a cake division $\{I_1, I_2, \ldots, I_n\}$ it is defined to be the sum of the values that the division generates among the agents, $\sum_i v_i(I_i)$.  Maximizing social welfare is a well-studied objective in resource-allocation contexts. In the cake-cutting setup, this maximization problem is known to be {\rm APX}-hard under general valuations \cite{arunachaleswaran2019fair}. Complementarily, if the value densities bear MLRP, then we can find (up to an arbitrary precision) a social welfare maximizing division with connected pieces in $\mathcal{O} \left( {\rm poly } \left(n, \log \lambda \right) \right)$ time (Theorem~\ref{theorem:SocialWelfare}). As mentioned previously, our algorithm for this problem is based on a dynamic program. \\

\noindent
\emph{Egalitarian Welfare:} The egalitarian (Rawlsian) welfare of a cake division $\{I_1, \ldots, I_n\}$ is defined as the value of the least well-off agent, i.e., $\min_i \  v_i (I_i)$. From a welfarist perspective, maximizing this minimum value among cake divisions with connected pieces is an important fairness objective. However, no nontrivial approximation guarantees are known for this problem under general valuations; the work of Aumann et al. \cite{aumann2013computing}  shows that maximizing egalitarian welfare across all contiguous cake divisions is {\rm APX}-hard. Complementing this hardness result, we develop an algorithm that, under MLRP, maximizes egalitarian welfare (up to an arbitrary precision) and runs in $\mathcal{O} \left( {\rm poly } \left(n, \log \lambda \right) \right)$ time (Theorem~\ref{theorem:Maxmin}).  

Our algorithm for maximizing egalitarian welfare is based on a ``moving-knife'' procedure. This procedure, for a given a target value $\tau >0$, iteratively selects points $x_0 = 0, x_1, x_2, \ldots, x_n \leq 1$ such that the each interval $[x_{i-1}, x_i]$ is of value $\tau$ to agent $i \in [n]$. Let $\tau^*$ denote the optimal egalitarian welfare in the given cake-division instance. The useful observation here is that this moving-knife procedure will succeed for all $\tau \leq \tau^*$. This follows from the fact that here the intervals are assigned (left to right) in the MLRP order\footnote{Note that the agents are indexed accordingly.} and this ordering is also satisfied by an egalitarian welfare maximizing (in particular, a Pareto optimal) division. Therefore, by performing a binary search with $\tau$, we can find a contiguous division with egalitarian welfare arbitrarily close to the optimal.  \\

\noindent
\emph{Nash Social Welfare:} A balance between social and egalitarian welfare is obtained by considering the Nash social welfare \cite{nash1950bargaining, kaneko1979nash}. This welfare objective is defined as the geometric mean of the agents' values. It is known that, in general, it is {\rm APX}-hard to find a contiguous cake division that maximizes Nash social welfare \cite{arunachaleswaran2019fair}. Under MLRP, however, the problem of maximizing Nash social welfare admits a fully polynomial-time approximation scheme (Theorem~\ref{theorem:NSW}). We obtain this result via a dynamic program that considers the agents in the MLRP order. \\

\noindent
{\bf Additional Related Work:} Recently, approximation algorithms---with both additive \cite{hollender2019contiguous} and multiplicative \cite{arunachaleswaran2019fair} approximation guarantees---have been developed for finding contiguous envy-free cake divisions. The work of Brânzei and Nisan \cite{branzei2017query} develops query complexity upper and lower bounds for computing approximately envy-free allocations.  In contrast to these results, the current work focuses on cake-division instances with MLRP and shows that in such settings arbitrarily low envy can be achieved among the agents.  

The work of Bei et al. \cite{bei2012optimal} also studies contiguous cake division and provides computational results for maximizing social welfare subject to proportional fairness. Under this fairness constraint each agent $i \in [n]$ must receive an interval of value at least $1/n$ times $i$'s total value for the cake. Bei et al. \cite{bei2012optimal} show that, if the value densities are linear, then this problem admits a fully polynomial-time approximation scheme (FPTAS). We note that every pair of linear densities bear MLRP and, hence, such value-density functions fall within the purview of the current work. However, our algorithm is incomparable to the FPTAS of Bei et al. \cite{bei2012optimal}--we focus on maximizing social welfare without the fairness constraints.  Also, envy-freeness is not addressed in \cite{bei2012optimal}. 

Another well-studied fairness notion is that of {equitability.} Specifically, a cake division $\{I_1, I_2, \ldots, I_n \}$ is said to be \emph{equitable} iff all the agents derive the same value from the intervals assigned to them, $v_i(I_i) = v_j(I_j)$ for all $i$ and $j$ \cite{dubins1961cut, alon1987splitting}. In other words, equitability ensures that all the agents are equally well-off. Cechl{\'a}rov{\'a} and Pill{\'a}rov{\'a} \cite{cechlarova2012computability} consider the computation of equitable cake divisions with connected pieces. They showed that---given access to ``reverse'' cutting queries---such divisions can be efficiently computed, up to an arbitrary precision. We note that value densities that satisfy MLRP have, by definition, full support over the cake. In such a case, the reverse cutting queries can be simulated by standard (cut) queries in the Robertson-Webb model. Hence, under MLRP, strong algorithmic results hold for equitability as well. 

Cake-division algorithms for specific classes of valuations have been studied in \cite{cohler2011optimal} and \cite{kurokawa2013cut}. The work of Kurokawa et al.~\cite{kurokawa2013cut} provides a query-efficient algorithm for envy-free, noncontiguous cake division under piecewise linear densities. Cohler et al.~\cite{cohler2011optimal} also address the noncontiguous version of the problem, and for piecewise constant densities they develop a polynomial-time algorithm that computes an envy-free division with optimal social welfare. In contrast to these results our focus is on contiguous cake division. 

\section{Notation and Preliminaries}
\label{section:notations}

This work studies the problem of dividing a cake $[0,1]$ among $n$ agents. Throughout, we will focus on a well-studied formulation of cake cutting which requires that each agent should receive a contiguous piece of the cake, i.e., the goal is to partition the cake $[0,1]$ into $n$ pairwise disjoint intervals and assign them among the $n$ agents in a fair/efficient manner.  

The cardinal preferences of each agent $i \in [n]$ is induced by a value-density function $f_i : [0,1] \mapsto \mathbb{R}_+$. Following standard conventions, we will assume that each value-density function $f_i$ is (Riemann) integrable. In particular, the (finite) integral of $f_i$ induces agent $i$'s valuation function over the intervals contained in $[0,1]$ (i.e., over the pieces of the cake): $v_i(I) \coloneqq \int_\ell^r  f_i(x) \ {d}x$ denotes the value that agent $i \in [n]$ has for any interval $I =[\ell, r] \subset [0,1]$. For notational convenience, we will write $v_i(a, b)$ to denote agent $i$'s value for interval $[a,b] \subseteq [0,1]$.

The integrability\footnote{Our results hold for integrable value densities and do not necessarily require $f_i$s to be continuous. Recall that, by definition, Riemann integrable functions are bounded. Also, every (bounded) continuous function on an (bounded) interval is Riemann integrable, but the converse is not true.} and nonnegativity of value-densities $f_i$ imply that the corresponding valuations $v_i$ are (i) nonnegative, (ii) divisible: for every interval $[\ell,r]$ and parameter $ \kappa \in [0,1]$, there exists a $z \in [\ell,r]$ with the property that $v_i(\ell,z) = \kappa \  v_i(\ell,r)$,\footnote{This implication can be obtained by applying the intermediate value theorem to the antiderivative of $f_i$.} and (iii) sigma additive: $v_i(I \cup J) = v_i(I) + v_i(J)$, for all disjoint intervals $I, J \subset [0,1]$. The divisibility property ensures that the valuations are non-atomic, i.e., $v_i([x,x]) = 0$ for all $i \in [n]$ and $x \in [0,1]$. Furthermore, this property allows us, as a convention, to regard two intervals to be disjoint even if they intersect exactly at an endpoint.

We additionally assume that the valuations are normalized such that the value of the entire cake is equal to one for every agent $i \in [n]$, i.e., $\int_0^1 f_i(x)dx = v_i(0,1) =1$. Hence, the value-densities $f_i$s constitute probability density functions over the cake $[0,1]$. {We note that this is a standard assumption in the cake-cutting framework and we conform to it for the purpose of brevity. All of our results hold true, even otherwise.} \\


  
\noindent
{\bf Problem Instances:} A \emph{cake-division instance} $\Cake$ is a tuple $\langle [n], \{f_i \}_{i \in [n] } \rangle$ where $[n]=\{1,2, \ldots, n\}$ denotes the set of $n \in \mathbb{Z}_+$ agents and $f_i$s denote the value-density functions of the agents. We will use the notation $(f_i)_{i \in [n]}$ to denote an ordered set of value-density functions of $n$ agents.\\ 

\noindent
{\bf Robertson-Webb Query Model:} While, for exposition, we specify $f_i$s as part of the problem instance, our algorithms only require oracle access to the valuations. In particular, the developed algorithms operate under the Robertson-Webb model~\cite{robertson1998cake}, which supports oracle access to agents' valuations in the form of \emph{eval} and \emph{cut} queries: \\
\noindent 
(i) \emph{Eval queries:} for each agent $i \in [n]$, we have (blackbox) access to function $\Eval_i: [0,1] \times [0,1] \mapsto \mathbb{R}_+$, which  when queried with any interval $[\ell, r]$ returns (in unit time) the value that agent $i$ has for this interval, i.e., $ \Eval_i (\ell, r) = v_i (\ell, r)$. \\
\noindent
(ii) \emph{Cut queries:} for each agent $i \in [n]$, we can also query function $\Cut_i: [0,1] \times \mathbb{R}_+ \rightarrow [0,1]$, which given an initial point $\ell \in [0,1]$ and a target value $\tau \in \mathbb{R}_+$, returns $\Cut_i(\ell, \tau) = y$ where $y \in [0,1]$ is the leftmost point with the property that $v_i(\ell, y) = \tau$. {If for a given $\ell \in [0,1]$ and $\tau \in \mathbb{R}_+$, there does not exist a $y \in [\ell,1]$ such that $v_i(\ell, y) = \tau$, then we have, by convention, $\Cut_i(\ell, \tau) = 1$.}  \\

\noindent
{\bf Allocations and Cake Divisions:} 
As mentioned previously, the goal is to assign each agent a single interval. Towards this end, for any cake-division instance with $n$ agents, we define an  \emph{allocation} to be a collection of $n$ pairwise-disjoint intervals, $\mathcal{I} = \{I_1, I_2, \ldots, I_n \}$, where interval $I_i$ is assigned to agent $i \in [n]$ and $\cup_{i \in [n]} \ I_i = [0,1]$. {Note that here the subscript of each interval identifies unique agent who has been assigned this interval.} In addition, we will refer to a collection of pairwise-disjoint intervals $\mathcal{J} = \{J_1, J_2, \ldots, J_n\}$ as a \emph{partial allocation} if they do not cover the entire cake, $\cup_{i=1}^n J_i  \subsetneq [0,1]$.

For an allocation $\mathcal{I}=\{I_1, \ldots, I_n\}$, the endpoints of the constituent intervals will be referred to as the \emph{cut points} of $\mathcal{I}$, i.e., if $I_i = [x_{i-1}, x_i]$ for $ 1 \leq i \leq n$, then the cut-points are $\{x_0 = 0, x_1, \ldots, x_n = 1\}$. 

We will throughout use the term allocation to specifically refer to partitions of the cake in which each agent receives a connected piece, i.e., receives exactly one interval. More generally, a \emph{cake division} will be used to denote partitions of the cake $\mathcal{D}=\{D_1, D_2, \ldots, D_n\}$ in which agent $i$ receives $D_i$, a finite collection of intervals. Here, the bundles $D_i$s are pairwise disjoint and their union covers the entire cake $[0,1]$. \\

\noindent
In this work we develop algorithmic results for the following notions of fairness and economic efficiency. 

\noindent
{\bf Envy-Freeness:} For a cake-division instance $\Cake$, an allocation $\mathcal{I} = \{I_1, \ldots, I_n \}$ is said to be \emph{envy-free} iif each agent prefers its own interval over that of any other agent, $v_i(I_i) \geq v_i(I_j)$ for all agents $i, j \in [n]$. \\

\noindent
{\bf Pareto Optimality:} Given a cake-division instance $\Cake$, a division $\mathcal{D}=\{D_1, \ldots, D_n \}$ is said to \emph{Pareto dominate} another division $\mathcal{C}=\{C_1, \ldots, C_n \}$ iff $v_i(D_i) \geq v_i(C_i)$ for all agents $i \in [n]$ and, there exists at least one agent $k \in [n]$ such that $v_k(D_k) > v_k(C_k)$. Consequently, a cake division is said to be \emph{Pareto optimal} iff it is not Pareto dominated by any other division. 

Recall that a cake division refers to a partition of the cake in which agent receives a finite collection of intervals. By contrast, in an allocation each agent receives a single interval. The algorithms developed in this work compute allocations. Interestingly, though, the Pareto optimality guarantees achieved by our algorithms are stronger in the sense that optimality holds across all cake divisions; specifically, under MLRP, we establish that particular allocations are Pareto optimal not only among the set of all allocations but also among all cake divisions. \\ 

\noindent 
{\bf Social Welfare:} Social welfare is a standard measure of collective value. Specifically, \emph{social welfare} for an allocation $\mathcal{I} = \{I_1, \ldots, I_n\}$ is defined to be sum of the agents' valuations, $\SW(\mathcal{I}) \coloneqq \sum_{i=1}^n v_i(I_i)$. \\

\noindent
{\bf Egalitarian (Rawlsian) Welfare:} For an allocation $\mathcal{I} = \{I_1, \ldots, I_n\}$, the \emph{egalitarian welfare} is defined as the minimum value achieved across the agents, $\RW(\mathcal{I}) \coloneqq \min_{i \in [n]} v_i(I_i)$. \\ 


\noindent
{\bf Nash Social Welfare:} For an allocation $\mathcal{I}=\{I_1, \ldots, I_n\}$, the \emph{Nash social welfare} is defined to be the geometric mean of the agents' valuations, $\NSW(\mathcal{I}) \coloneqq \big( \prod_{i=1}^n v_i(I_i) \big)^{1/n}$. \\ 

Finding allocations that maximize the above-mentioned welfare notions is known to be {\rm APX}-hard, in general; see, e.g.,~\cite{aumann2013computing, arunachaleswaran2019fair}. Complementing these negative results, a key contribution of this work is to identify a broad class of cake-division instances that admit strong algorithmic results for these welfare objectives and envy-freeness. Specifically, we focus on value densities (distributions) that satisfy the {monotone likelihood ratio property} (MLRP). We will next define this property and note that our results hold for multiple distribution families that satisfy MLRP. \\






\noindent
{\bf Monotone Likelihood Ratio Property:} 
Probability density functions $f_i$ and $f_j$ (in order) are said to satisfy the \emph{monotone likelihood ratio property} (MLRP) iff, for every $x \leq y$ in the domain, we have $\nicefrac{f_j(x)}{f_i(x)}  \leq \nicefrac{f_j(y)}{f_i(y)}$.
That is, the {likelihood ratio} $\nicefrac{f_j(x)}{f_i(x)}$ is non-decreasing in the argument $x \in \mathbb{R}$. 


Note that MLRP does not require $f_i$ and $f_j$ to be monotonic (or unimodal) by themselves. We also observe that MLRP is transitive: if two pairs of distributions $(f_i, f_j)$ and $(f_j, f_k)$ satisfy MLRP separately, then the pair $(f_i, f_k)$ also conforms to MLRP. Furthermore, this property continues to hold under positive scaling: if $f_i$ and $f_j$ satisfy MLRP, then so do $\gamma_i f_i$ and $\gamma_j f_j$, for any positive scalars $\gamma_i, \gamma_j \in \mathbb{R}_+$. This fact, in particular, allows us to restrict MLRP densities (which are typically defined over the real line) to the cake $[0,1]$ and, at the same time, assume normalization $\int_{0}^1 f_i = 1$. 

In the cake-division context, we will say that a given collection $\{f_i\}_{i \in [n]}$ of value-density functions satisfies the \emph{monotone likelihood ratio property} iff there exists an order $\pi: [n] \rightarrow [n]$, among the $f_i$s, such that, for all $i \in [n-1]$, the consecutive likelihood ratios $\frac{f_{\pi (i+1)} \ (x)}{f_{\pi(i)} \ (x)}$ are non-decreasing in $x \in [0,1]$. That is, for each $i \in [n-1]$, the densities $f_{\pi(i)}$ and $f_{\pi(i+1)}$ bear MLRP over $[0,1]$. 

We will refer to this order $\pi$ as the {MLRP order} of the value densities. Lemma~\ref{MLRPorder} (in Appendix~\ref{appendix:find-mlrp-order}) shows that, given a cake-division instance in the Robertson-Webb query model, with the promise that the underlying value densities satisfy MLRP (i.e., given a promise problem), one can efficiently find the MLRP order $\pi$. Hence, without loss of generality, we will throughout assume that the $n$ agents are indexed such that $\pi$ is the identity permutation, i.e., for all $i \in [n-1]$, the likelihood ratio $\frac{f_{i+1} (x)}{f_i(x)}$ is non-decreasing in $x \in [0,1]$. 

It is relevant to note that, to be well defined, MLRP requires the value densities $f_i$s to be strictly positive over the cake $[0,1]$. Hence, for cake-division instances $\langle [n], \{f_i\}_{i \in [n]} \rangle$ with MLRP, we have $f_i(x) >0$ for all $i \in [n]$ and $x \in [0,1]$. \\

\noindent
{\bf Instantiations of MLRP:} 
MLRP induces a total order on linear value densities $f_i(x) = a_i x + b_i$; see Appendix~\ref{appendix:mlrp-use-cases} for details. Hence, our results imply that if, in a cake-division instance, the value densities of all the agents are linear, then an envy-free (or welfare-maximizing) allocation can be computed efficiently. 

Many other distribution families are also known to bear MLRP, e.g., Gaussian distributions (with the same variance), Poisson distributions, and single-parameter exponentials. Therefore, our algorithmic results address, in particular, cake-division instances wherein all the agents have Gaussian value densities with the same variance, but different means. 





In fact, it is known that linear translations of any log-concave function $g$---i.e., densities of the form $f_{\theta}(x) \coloneqq g(x-\theta)$, for $\theta \in \mathbb{R}$---satisfy MLRP~\cite{saumard2014log}. Hence, our results also hold for linear translations of the following (log-concave) distributions: Laplace, uniform, multivariate Gaussian, gamma, beta, Subbotin, chi-square, Dirichlet, and logistic. 


These instantiations substantiate the applicability of our algorithmic results which, through MLRP, address a wide range of cake-division instances. \\

\noindent
{\bf Lipschitz Constant of {Cut} and {Eval} Queries:} We say that the cut and eval queries in a cake-division instance are $\lambda$-Lipschitz iff the following inequalities hold for each agent $i \in [n]$:
\begin{align*}
|\Eval_i(\ell',r') - \Eval_i(\ell,r)|  & \leq \lambda \ \| (\ell',r') - (\ell,r) \|_{\infty}  \ \ \quad \text{for all} \ \ (\ell',r'), (\ell,r) \in [0,1] \times [0,1]\\
|\Cut_i(\ell',\tau') - \Cut_i(\ell,\tau)| & \leq \lambda \  \| (\ell',\tau') - (\ell,\tau) \|_{\infty}  \ \ \quad \text{for all} \ \ (\ell',\tau'), (\ell,\tau) \in [0,1] \times \mathbb{R}_+
\end{align*}

A useful consequence of MLRP (and the integrability of the value densities) is that the corresponding {cut} and {eval} queries are in fact $\lambda$-Lipschitz, for a finite $\lambda \geq 1$;\footnote{If a function is $\lambda'$-Lipschitz then, it is $\lambda$-Lipschitz as well, for all $\lambda \geq \lambda'$.} see Proposition~\ref{Lip}. This proposition follows from the fact that (Riemann) integrable value densities $f_i$s are, by definition, bounded. Furthermore, as mentioned previously, MLRP mandates that the value densities are strictly positive over the cake. Therefore, for each agent $i \in [n]$ and $x \in [0,1]$ we have $0 < L \leq f_i(x) \leq U$ for some $L, U \in \mathbb{R}_+$. Proposition~\ref{Lip} asserts that the Lipschitz constant $\lambda$ can be expressed in terms of these bounding parameters, $\lambda \leq \max\{1/L, U, U/L\}$.  



It is worth pointing out that besides MLRP (and, hence, the positivity of the value densities), all the other assumptions made in this work are standard. 

The time complexities of our algorithms depend polynomially on the the bit complexity of the Lipschitz constant $\lambda \geq 1$.\footnote{In the case of linear value densities, $f_i(x) = a_i x + b_i$, the bit complexity of the Lipschitz constant $\lambda$ is proportional to the bit complexity of the coefficients $a_i$s and $b_i$s (Proposition~\ref{Lip}). Hence, if linear densities are explicitly given as input, then we have a polynomial (in the input size) runtime bound.} As mentioned previously, in general, such a runtime dependency on $\log \lambda$ is unavoidable; Appendix~\ref{appendix:example-precision-loss} provides an illustrative examples. Furthermore, it is possible---even with rational and MLRP value densities---that the exact envy-free/welfare-maximizing allocations are induced by irrational cuts (Appendix~\ref{appendix:example-precision-loss}). That is, in general, one cannot expect an efficient algorithm that outputs an {exact} envy-free (or welfare-maximizing) allocation. Therefore, when considering efficient algorithms for cake division, a precision loss is inevitable. However, our algorithms ensure that this loss is arbitrarily small. Specifically, given a cake-division instance $\Cake$ in which the value densities satisfy MLRP, we can find, in time that is polynomial in $\log (1/\eta)$ (along with $n$ and $\log \lambda$), an envy-free allocation $\mathcal{I} = \{I_1,\ldots, I_n \}$ such that, $v_i(I_i) \geq v_i(I_j) - \eta$ for all agents $i, j \in [n]$. Since the precision parameter $\eta$ can be driven exponentially close to zero in polynomial (in the bit complexity of $\eta$) time, we will say that an envy-free allocation can be computed {up to an arbitrary precision}. 


Similarly, in the context of maximizing welfare (social or egalitarian) welfare, given a cake-division instance wherein the value densities bear MLRP, we can find---in time that is polynomial in $ \log(1/\eta)$---an allocation with the (social or egalitarian) welfare $\eta$ (additively) close to the optimal. Hence, as in the case of envy-freeness, we assert that a welfare-maximizing allocation can be computed efficiently, up to an arbitrary precision.

\section{Main Results}
\label{section:MainResults}
This section presents the statements of our key results.  \\

\noindent
\textbf{Envy-Freeness:} In Section~\ref{EFdivisions} we prove that for cake-division instances, in which the value densities satisfy MLRP, the problem of finding an envy-free allocation {essentially admits a polynomial-time algorithm}. 



 \begin{restatable}{theorem}{EFdivision}
	\label{theorem:EFdivision}
	Let $\Cake = \langle [n], (f_i )_{i \in [n] } \rangle $ be a cake-division instance in which the value-density functions satisfy the monotone likelihood ratio property. Then, in the Robertson-Webb query model, an envy-free allocation of $\Cake$ can be computed (up to an arbitrary precision) in $\mathcal{O}\left( {\rm poly} ( n, \log \lambda ) \right)$ time; here $\lambda \in \mathbb{R}_+$ is the Lipschitz constant of the cut and eval queries.
\end{restatable} 

Recall that, under MLRP, the cut and eval queries are necessarily $\lambda$-Lipschitz (Proposition~\ref{Lip}). \\

\noindent
\textbf{Pareto Optimality:} Weller's theorem~\cite{weller1985fair} asserts that there always exists some cake division (though, not necessarily with connected pieces) which is both envy-free and Pareto optimal (among all cake divisions, with or without connected pieces). We show that, in the context of MLRP,  \emph{every} envy-free allocation is in fact Pareto optimal (among all cake divisions). Therefore, for cake-division instances with MLRP, the allocation computed by our algorithm is not only envy-free (fair) but also Pareto optimal, up to an arbitrary precision. 



 \begin{restatable}{theorem}{EFPO} 
	\label{theorem:ef-po}
		Let $\Cake$ be a cake-division instance wherein the value-density functions satisfy the monotone likelihood ratio property. Then, every envy-free allocation in $\Cake$ is also Pareto optimal (over the set of all cake divisions).
\end{restatable}

	

\noindent
\textbf{Social Welfare:} In Section \ref{section:social welfare} we show that, up to an arbitrary precision, a social welfare maximizing allocation can be computed efficiency under MLRP.

\begin{restatable}{theorem}{SocialWelfare}
	\label{theorem:SocialWelfare}
Let $\Cake =\langle [n], (f_i)_{i \in [n]} \rangle$ be a cake-division instance in which the value-density functions satisfy the monotone likelihood ratio property. Then, in the Robertson-Webb query model, an allocation that achieves the optimal social welfare in $\Cake$ can be computed (up to an arbitrary precision) in $\mathcal{O}\left( {\rm poly} ( n, \log \lambda ) \right)$ time; here $\lambda \in \mathbb{R}_+$ is the Lipschitz constant of the cut and eval queries.
\end{restatable}

\noindent 
\textbf{Egalitarian Welfare:} Section~\ref{section:max-min} addresses the problem of maximizing egalitarian welfare. Specifically, we prove that, in cake-division instances with MLRP, an allocation with egalitarian welfare arbitrarily close to the optimal can be computed efficiently. 

 \begin{restatable}{theorem}{Maxmin}
	\label{theorem:Maxmin}
	Let $\Cake =\langle [n], (f_i)_{i \in [n]} \rangle$ be a cake-division instance in which the value-density functions satisfy the monotone likelihood ratio property. Then, in the Robertson-Webb query model, an allocation that achieves the optimal egalitarian welfare in $\Cake$ can be computed (up to an arbitrary precision) in $\mathcal{O}\left( {\rm poly} ( n, \log \lambda ) \right)$ time; here $\lambda \in \mathbb{R}_+$ is the Lipschitz constant of the cut and eval queries.	
\end{restatable} 


\noindent
{\bf Nash Social Welfare:} In Section~\ref{section:nsw} we show that, under MLRP, the problem of maximizing Nash social welfare admits a fully polynomial-time approximation scheme (FPTAS). 


\begin{restatable}{theorem}{NashSocialWelfare}
\label{theorem:NSW}
Let $\Cake =\langle [n], (f_i)_{i \in [n]} \rangle$ be a cake-division instance in which the value-density functions satisfy the monotone likelihood ratio property. Then, in the Robertson-Webb query model, an allocation with Nash social welfare at least $(1- \varepsilon)$ times the optimal (Nash social welfare) in $\Cake$ can be computed in time that is polynomial in $1/\varepsilon$, $\log \lambda$, and  $n$; here $\lambda \in \mathbb{R}_+$ is the Lipschitz constant of the cut and eval queries.	
\end{restatable}

\section{Implications of the Monotone Likelihood Ratio Property}

This section provides useful implications of MLRP (Lemma~\ref{lemma:R2R}). This result will be used in subsequent sections towards the analysis of our algorithms.

\begin{restatable}{lemma}{RacetotheRatios}
	\label{lemma:R2R}
Let $f_i$ and $f_j$ be two (ordered) value-density functions that bear MLRP i.e., the likelihood ratios satisfy $\frac{f_{j} (b)}{f_i(b)} \leq \frac{f_{j}(c)}{f_i(c)}$ for all $0 \leq b \leq c \leq 1$. Then, $f_i$ and $f_j$ satisfy the following two properties
\begin{itemize}
	\setlength\itemsep{0.001em}
\item[(i)] The values of the intervals satisfy $\frac{\int \limits_a^b f_{j}}{\int \limits_a^b f_{i}} \leq \frac{\int \limits_c^d f_{j}}{\int \limits_c^d f_{i}}$ for all $[a,b], [c,d] \subseteq [0,1]$ with $ b \leq c$.
\item[(ii)] The (normalized) values satisfy $\frac{\int \limits_x^b f_i}{\int \limits_a^b f_i} \leq \frac{\int \limits_x^b f_{j}}{\int \limits_a^b f_{j}} $ for all intervals $[a,b] \subseteq [0,1]$ and all $x \in [a,b]$.
\end{itemize}	
Moreover, properties (i) and (ii) are equivalent.  
\end{restatable} 
Here, if the likelihood ratio $\frac{f_{j}(x)}{f_i(x)}$ is strictly increasing, then we have a strict inequality in the corresponding implications.

The proof of this lemma is delegated to Appendix~\ref{appendix:proof-race-to-ratio}. Note that, in terms of the agents' valuations $v_i$ and $v_j$, property (i) in Lemma~\ref{lemma:R2R} can be expressed as
$\frac{v_{j} (a,b)}{ v_i(a,b) } \leq \frac{v_{j}(c,d)}{v_i(c,d)}$ for all $[a,b], [c,d] \subseteq [0,1]$ with $ b \leq c$. 
Similarly, property (ii) corresponds to $\frac{v_{i} (x,b)}{ v_i(a,b) } \leq \frac{v_{j}(x,b)}{v_{j}(a,b)}$ for all $[a,b] \subseteq [0,1]$ and all $ x \in [a,b]$. 

It is well-known that MLRP implies first-order stochastic dominance (see Lemma~\ref{SD}). Interestingly, property (ii) provides a strengthening: over any interval $[a,b]$, the normalized (by $v_{j}(a, b)$) values of agent $j$ stochastically dominate the normalized (by $v_{i}(a, b)$) values of agent $i$.



\section{Envy-Freeness} \label{EFdivisions}
In this section we develop an efficient algorithm for finding envy-free allocations in cake division instances with MLRP (Theorem~\ref{theorem:EFdivision}). Towards this goal, we define a class of cake divisions, referred to as \emph{ripple divisions} (Definition~\ref{defn:RD}), and prove that, under MLRP, every ripple division induces an envy-free allocation (Theorem~\ref{theorem:RD-EF}). Existential and computational guarantees for {ripple divisions} are established in Section~\ref{section:rd-exist} (Lemma~\ref{RDexistence}) and Section~\ref{section:rd-compute} (Lemma~\ref{RDcomputation}), respectively. Section~\ref{proofthm1} builds upon these results to prove our main result (Theorem~\ref{theorem:EFdivision}) for envy-freeness.  

We establish the universal existence of ripple divisions through the intermediate value theorem, i.e., a one-dimensional fixed-point argument (Lemma~\ref{RDexistence}). Consequentially, for cake-division instances with MLRP, we develop an alternate proof of existence of envy-free allocations. Since one can use binary search to find fixed points in the one-dimensional setting, this proof in fact leads to an algorithm for finding ripple divisions and, hence, envy-free divisions. 


\begin{definition}[Ripple Division]  \label{defn:RD}
	Given a cake-division instance $\Cake = \langle [n], \{f_i \}_{i} \rangle $, a collection of points $x^*_0=0 \leq x^*_1 \leq \dots \leq x^*_{n-1} \leq x^*_n = 1$ is said to form a ripple division of the cake iff 
	\begin{align*}
	  v_{i}(x^*_{i-1}, x^*_{i}) &= v_{i}(x^*_{i}, x^*_{i+1})>0 \ \ \text{for all agents} \  i \in [n-1].
	\end{align*}
\end{definition}


For establishing existence of ripple divisions, we first consider a relaxation of Definition~\ref{defn:RD} wherein do not enforce the last cut point (i.e., $x^*_n$) to be equal to one. Under this relaxation, the intervals $\{[x_{i-1}, x_i]\}_{i=1}^n$ do not cover the entire interval $[0,1]$ (instead, they span $[0,x_n]$) and, hence, lead to a partial allocation of the cake. 

Also, by convention, the agents are indexed following the MLRP order: for each $i \in [n-1]$, the likelihood ratio $f_{i+1}/f_i$ is nondecreasing. Hence, assigning interval $[x_{i-1}, x_i]$ to agent $i \in [n]$ provides an allocation wherein the intervals are assigned (left to right on the cake) in accordance with the MLRP order. We will show that the (partial) allocation obtained by assigning interval $[x_{i-1}, x_i]$ to agent $i \in [n]$ is envy-free (Theorem~\ref{theorem:RD-EF}). 

\begin{definition}[$\delta$-Ripple Division]  \label{defn:deltaRD}
	Given a cake-division instance $\Cake = \langle [n], \{f_i \}_{i} \rangle $, a collection of points $x_0=0 \leq x_1 \leq \dots \leq x_{n-1} \leq x_n \leq 1$ is said to form a $\delta$-ripple division of the cake iff $x_n \geq 1 - \delta$ and 
	\begin{align*}
	  v_{i}(x_{i-1}, x_{i}) &= v_{i}(x_{i}, x_{i+1})>0 \ \ \text{for all agents} \  i \in [n-1].
	\end{align*}
\end{definition}

Both Definitions~\ref{defn:RD} and~\ref{defn:deltaRD} require that, for all $i \in [n-1]$, agent $i$'s value for the $i$th interval ($[x^*_{i-1}, x^*_{i}]$ and $[x_{i-1}, x_{i}]$, respectively) is positive.  Also, note that a $0$-ripple division is an exact ripple division.

To compose the value equalities that define a $\delta$-ripple division (Definition~\ref{defn:deltaRD}), we consider $(n-1)$ functions $\RD_i: [0,1] \mapsto [0,1]$, for $2 \leq i \leq n$. In particular, focusing on the equalities considered in Definition~\ref{defn:deltaRD} (i.e., $v_i(x_{i-1}, x_i) = v_i(x_i, x_{i+1})$), we note that that if we set the first cut point $x_1=x \in [0,1]$, then all the subsequent points $x_2, \ldots, x_{n}\in [0,1]$ are fixed as well. In particular, $x_2$ is the point that satisfies $v_1(0, x) = v_1(x, x_2)$ and, iteratively, $x_{i+1}$ can be identified from $v_i(x_{i-1}, x_i) = v_i(x_i, x_{i+1})$. The functions $\RD_i$s capture this ``ripple'' effect and can be expressed as compositions of cut and eval queries. Formally, the functions $\RD_i : [0,1] \mapsto [0,1]$, for $i \in \{2,3, \dots, n\}$, are recursively defined as follows\footnote{The first cut point $x$ is specified upfront and, hence, we do not require $\RD_1$. The functions $\RD_i$s are defined for $i \in \{2, \ldots, n\}$.} 
\begin{align} \label{RDfunction}
\RD_2(x) &\coloneqq \Cut_1\left(x, \Eval_1(0,x)\right) \nonumber\\
\RD_i(x) &\coloneqq \Cut_{i-1} \Big( \RD_{i-1}(x), \Eval_{i-1} \big( \RD_{i-2}(x), \RD_{i-1}(x) \big) \Big) \quad \text{ for } i \in \{3,4, \ldots, n\}
\end{align}

In particular, $\RD_n(x)$ denotes the value of the last cut point $x_n$ obtained by initializing the ripple effect with $x_1 = x$. Also, note that $\RD_n(0) = 0$ and $\RD_n(1) =1$; by convention, the response to a cut query $\Cut_i( \ell, \tau)$ is truncated to $1$ iff $\tau$ is greater than the entire value to the right of $\ell$, i.e., if $v_i(\ell, 1) < \tau$.

Since $\RD_i$s can be expressed as a composition of the cut and eval queries, these functions can be efficiently computed in the Robertson-Webb query model. Moreover, using the fact that the cut and eval queries are $\lambda$-Lipschitz, in the following proposition we establish that $\RD_n$ is also Lipschitz continuous.

\begin{restatable}{proposition}{PropositionCcnLipschitz}
If in a cake-division instance $\Cake = \langle [n], \{f_i \}_{i} \rangle $ the cut and eval queries are $\lambda$-Lipschitz, then the function $\RD_n$ is $\lambda^{2(n-1)}$-Lipschitz.
 \label{RDn}
\end{restatable}

The proof of Proposition~\ref{RDn} is deferred to Appendix~\ref{appendix:exist-compute-rd}. The Lipschitz continuity of $\RD_n$ turns out to be a key property, both for establishing the existence of ripple divisions and in developing an efficient algorithm for finding them. Specifically, we can apply the intermediate value theorem to $\RD_n$ and prove that, for any $\delta >0$, there exists a $\delta$-ripple division; see Lemma~\ref{lemma:delta-rd-exist} below. The universal existence of $\delta$-ripple divisions, for all $\delta >0$, along with a limit argument ($\delta \to 0$), establishes our existential result (Lemma~\ref{RDexistence}) for ripple divisions, i.e., for $0$-ripple divisions.

The Lipschitzness of $\RD_n$ also ensures that, via binary search, we can find a $\delta$-ripple division (via a point $x_1 \in (0,1)$ that satisfies $\RD_n(x_1) \geq 1 - \delta$ ), in time that is polynomial in $\log \left( \frac{1}{\delta} \right)$. This runtime dependence ensues that the precision parameter $\delta$ can be driven exponentially close to zero, in polynomial (in the bit complexity of $\delta$) time; see Lemma~\ref{RDcomputation}.

\subsection{Existence of Ripple Divisions}
\label{section:rd-exist}

\begin{lemma} 
\label{lemma:delta-rd-exist}
Let $\Cake = \langle [n], \{f_i \}_i \rangle$ be a cake-division instance, with $\lambda$-Lipschitz cut and eval queries, and let parameter $\delta \in (0,1)$. Then, in $\Cake$, there always exist a point $\widehat{x} \in (0,1)$ with the property that $\RD_n\left(\widehat{x}\right) = 1- \delta$.

Furthermore, such a point $\widehat{x}$ initializes a $\delta$-ripple division: the points $x_0=0$, $x_1 = \widehat{x}$, and $x_i= \RD_i\left( \widehat{x} \right)$, for $2 \leq i \leq n$, form a $\delta$-ripple division in $\Cake$. Here, the functions $\RD_2, \ldots, \RD_n$ are defined with respect to the cut and eval queries of $\Cake$. 
\end{lemma}

\begin{proof} The function $\RD_n$ is continuous (Proposition \ref{RDn}) and it satisfies $\RD_n(0) = 0$ along with $\RD_n(1) =1$. Hence, the intermediate value theorem ensures that, for any $\delta \in (0,1)$, there must exist a point $\widehat{x} \in (0,1)$ which satisfies $\RD_n\left(\widehat{x} \right) = 1-\delta$. This shows that the required point $\widehat{x}$ always exists. 

We will complete the proof by establishing that the points $x_0=0$, $x_1 = \widehat{x}$, and $x_i= \RD_i\left( \widehat{x} \right)$, for $2 \leq i \leq n$, form a $\delta$-ripple division. Note that, since $\delta >0$, we have $\RD_n\left(\widehat{x} \right)  <1$. This strict inequality implies that the intermediate points $x_i$s were \emph{not} truncated to one. Indeed, in this case, all of the value equalities in Definition~\ref{defn:deltaRD} hold, i.e., $v_i(x_{i-1}, x_i) = v_i(x_i, x_{i+1})$ for all $i \in [n-1]$. It remains to show that these values are positive. 

The bound $\delta <1$ gives us $\RD_n \left( \widehat{x} \right) >0 = \RD_n(0)$. Therefore, we have $\widehat{x}>0$, i.e., $x_1 = \widehat{x} > 0 = x_0$. In the current setting, all nonempty intervals have positive value for the agents. This follows from the continuity of the cut queries.\footnote{In fact, under MLRP, we explicitly have $f_i(x) >0$ for all $x \in [0,1]$, to have the likelihood ratios be well-defined.} Since the interval $[0, x_1]$ is nonempty, agent $1$ has a positive value for it, $v_1(0, x_1) >0$.  Furthermore, the value equality $v_1(x_1, x_2) = v_1(0, x_1)$ gives us $v_1(x_1, x_2)  > 0$, i.e., $x_2>x_1$.  Extending this argument iteratively shows that $x_0 < x_1 < \ldots < x_n <1$. In other words, each agent receives a positive value under the cut points, i.e., $v_i(x_{i-1}, x_i) = v_i(x_i, x_{i+1}) >0$ for all $i \in [n-1]$. Hence, $x_i$s form a $\delta$-ripple division.
\end{proof}

Next we use Lemma~\ref{lemma:delta-rd-exist} and a limit argument ($\delta \to 0$) to establish universal existence of ripple divisions. 

\begin{restatable}{lemma}{RDexistence}
		\label{RDexistence}
	Let $\Cake$ be a cake-division instance in which the cut and eval queries are $\lambda$-Lipschitz. Then, $\Cake$ necessarily admits a ripple division.
\end{restatable}

 \begin{proof}
Lemma~\ref{lemma:delta-rd-exist} asserts that, for any $\delta \in (0,1)$, there exists a collection of points $x^{\delta}_0=0 < x^{\delta}_1 < \ldots < x^{\delta}_{n-1} < x^{\delta}_n = 1-\delta$ that form a $\delta$-ripple division. Note that here $x^\delta_n = 1- \delta$. We consider the sequence of these $\delta$-ripple divisions, $\mathcal{S} \coloneqq  \left\langle (x^{\delta}_0, x^{\delta}_1, \ldots, x^{\delta}_n) \right\rangle_{\delta \in (0,1)}$.
 
Since $\mathcal{S}$ is a nonempty and bounded sequence in $[0,1]^n$, the Bolzano-Weierstrass theorem ensures that $\mathcal{S}$ contains a convergent subsequence, say $\left\langle ({x}^{\delta_j}_0, {x}^{\delta_j}_1, \dots, {x}^{\delta_j}_n)\right\rangle_{\delta_j>0}$. Write $(x^*_0, x^*_1, \dots, x^*_n) \in [0,1]^n$ to denote the limit of this subsequence as $\delta_j$ tends to zero. We will show that the points $x^*_0, x^*_1, \ldots, x^*_n$ form a ripple division in $\Cake$ (see Definition~\ref{defn:RD}), i.e., establish that $x^*_n =1$ and $v_i(x^*_{i-1}, x^*_i) = v_i(x^*_i,x^*_{i+1})>0$ for all agents $i \in [n-1]$.

First, note that $x^*_n =1$, since the sequence $\langle 1-\delta_j \rangle \to 1$ as $\delta_j \to 0$. Also, $x^*_0 =0$, since the constant sequence $\langle 0\rangle$ tends to $0$. 

We will next prove that  $v_{i}(x^*_{i-1}, x^*_{i}) = v_{i}(x^*_{i}, x^*_{i+1})$ for all agents $i \in [n-1]$. Given that the collection $\left( {x}^{\delta}_0, {x}^{\delta}_1, \dots, {x}^{\delta_j}_n \right)$ forms a $\delta_j$-ripple division, we have $v_{i}({x}^{\delta_j}_{i-1}, {x}^{\delta_j}_{i}) = v_{i}({x}^{\delta_j}_{i}, {x}^{\delta_j}_{i+1})$, for all $i \in [n-1]$. 

Recall that the value-density function $v_i$ (equivalently, $\Eval_i$) is continuous over $[0,1]^2$. Therefore, the sequential criterion of continuity\footnote{For any continuous function $g: [0,1]^2 \mapsto \mathbb{R}$, if a sequence $\langle a_t\rangle_t \in [0,1]^2$ converges to some $a \in [0,1]^2$, then the sequence $\langle g(a_t) \rangle_t$ must converge to $g(a)$.} gives us $\langle v_{i}({x}^{\delta_j}_{i-1}, {x}^{\delta_j}_{i}) \rangle \to v_i(x^*_{i-1}, x^*_i)$ and $\langle v_{i}({x}^{\delta_j}_{i}, {x}^{\delta_j}_{i+1}) \rangle \to v_i(x^*_i, x^*_{i+1})$, as $\delta_j$ tends to $0$. Applying the algebra of limits\footnote{For any two convergent sequences $\langle a_t \rangle_t \to a$ and $\langle b_t \rangle_t \to b$, if we have $a_t = b_t$ for all $t \geq 1$, then their limits must be equal as well, i.e., $a=b$.} on the two sequences $\langle v_{i}({x}^{\delta_j}_{i-1}, {x}^{\delta_j}_{i})\rangle$ and $\langle v_{i}({x}^{\delta_j}_{i}, {x}^{\delta_j}_{i+1})\rangle$, we obtain that their limits must be equal as well, i.e., $v_i(x^*_{i-1}, x^*_i) = v_i(x^*_i, x^*_{i+1})$.
 
 We now complete the proof by showing that these two equal values must be positive, $v_i(x^*_{i-1}, x^*_i) = v_i(x^*_i,x^*_{i+1})>0$. Since $x^*_n = 1$, we have $x^*_1>0$; otherwise, if $x^*_1 = 0$, then following the value equalities we would have $x^*_n = \RD_n(0) = 0$. In the current setting, all nonempty intervals have positive value for the agents. This follows from the Lipschitz continuity of the cut queries. Hence, $x^*_1> x^*_0 = 0$ gives us $0 < v_1(x^*_0, x^*_1) = v_1(x^*_1, x^*_2)$. This bound also implies $x^*_2 > x^*_1$.  Extending this argument iteratively shows that $0= x^*_{0} < x^*_{1} < \ldots < x^*_n =1$. Hence, the agents' values are positive, $v_i(x^*_{i-1}, x^*_i) >0$.  Overall, we get that the set of points $0=x^*_0<x^*_1< \dots < x^*_n=1$ form a ripple division in $\Cake$. 
\end{proof}
\subsection{Computation of Ripple Divisions}
\label{section:rd-compute}
In this section we present an efficient algorithm for computing $\delta$-ripple divisions. 
Lemma~\ref{lemma:delta-rd-exist} implies that the problem of computing a $\delta$-ripple division reduces to finding a point $x \in (0,1)$ that satisfies $\RD_n(x) \in [1- \delta, 1)$. The algorithm \textsc{BinSearch} (Algorithm~\ref{alg:BinSearch}) finds such a point $x$ (and, hence, a $\delta$-ripple division) via binary search. It is well-known that binary search can be used to compute fixed points in the one-dimensional setting. We provide the relevant details here for completeness. 

Specifically, in \textsc{BinSearch}, we initialize $\ell = 0$ along with $r =1$ and keep bisecting the interval $[\ell, r]$ as long as $\RD_n(\ell) < 1 - \delta$ and $\RD_n(r) =1$. {Recall that $\RD_n(0) = 0$ and $\RD_n(1) = 1$.}  

Since $\RD_n$ is Lipschitz continuous (Proposition~\ref{RDn}), the intermediate value theorem (applied on $[\ell, r]$ with the bounds $\RD_n(\ell) < 1 - \delta$ and $\RD_n(r) =1$) guarantees that in each considered interval $[\ell, r]$ there exists a point $x \in (\ell, r)$ which satisfies $\RD_n(x) = 1- \delta$.\footnote{Note  that this argument (in particular, the use of intermediate value theorem) does not require $\RD_n$ to be monotonic. Also, recall that $\RD_i$s can be computed efficiently in the Robertson-Webb query model. Hence, evaluating $\RD_i$s in \textsc{BinSearch} leads to at most a polynomial overhead in its runtime.} In each iteration of this algorithm the length of the considered interval $[\ell, r]$ reduces by a multiplicative factor of two and, hence, after an appropriate number of iterations, the required point $x$ and the midpoint of the interval $(\ell + r)/2$ get sufficiently close. In such a case, one can show (using the Lipschitz continuity of $\RD_n$) that the midpoint $(\ell + r)/2$ itself initializes a $2 \delta$-ripple division. Appendix~\ref{section:proofRDc} formalizes this runtime analysis and provides a proof of the following lemma. 

	\begin{algorithm}
		{ \small
			{\bf Input:} A cake-division instance $\Cake = \langle [n], \{f_i \}_{i} \rangle$, in the Robertson-Webb query model, and parameter $\delta >0$. \\
			{\bf Output:} A $\delta$-ripple division $0=x_0 \leq x_1 \leq x_2, \ldots \leq x_n \leq 1$.
			\caption{\textsc{BinSearch}} 
			\label{alg:BinSearch}
			\begin{algorithmic}[1]
			
				\STATE Initialize $\ell = 0$ and $r=1$
				\WHILE{$\ell < r$} 
				\IF {$ \RD_n\left(  \frac{\ell+r}{2} \right) < 1 - \delta$} 
				\STATE Update $\ell \leftarrow (\ell+r)/2$ \label{step:update-left} 
				\ELSIF {$ \RD_n\left(  \frac{\ell+r}{2} \right) =1$} 
				\STATE Update $r \leftarrow (\ell+r)/2$ \label{step:update-right} 
				\ELSIF {$ \RD_n\left(\frac{\ell+ r}{2} \right) \in [1 - \delta, 1) $}
				\STATE Set $x_0 = 0$, $x_1 = \frac{\ell+r}{2}$, and $x_i = \RD_i(x_1)$ for all $i \in \{2, \ldots, n\}$ 
				\RETURN the collection of points $x_0, x_1, \ldots, x_n$
				\ENDIF
				\ENDWHILE 
			\end{algorithmic}
		}
	\end{algorithm}

\begin{restatable}{lemma}{RDcomputation} \label{RDcomputation}
Let $\Cake = \langle [n], \{f_i \}_{i} \rangle $ be a cake-division instance in which the cut and eval queries are $\lambda$-Lipschitz. Then, for any $\delta \in (0,1)$ and in the Robertson-Webb query model, a $\delta$-ripple division of $\Cake$ can be computed in $\mathcal{O}\left( {\rm poly} ( n, \log \lambda, \log \frac{1}{\delta}) \right)$ time.
\end{restatable}

\subsection{From Ripple Divisions to Envy-Free Allocations}

The next theorem establishes the crucial connection between ripple divisions and envy-freeness. In particular, setting $\delta = 0$ in this theorem, we obtain that, under MLRP, every (exact) ripple division induces an envy-free allocation. 

\begin{restatable}{theorem}{RD-EF}
	\label{theorem:RD-EF}
Let  $\Cake = \langle [n], \{f_i \}_{i} \rangle$ be  a cake-division instance in which the value densities satisfy the monotone likelihood ratio property and let parameter $\delta \geq 0$. Then, every $\delta$-ripple division, $x_0=0 \leq x_1 \leq \dots \leq x_{n-1} \leq x_n \leq 1$, in $\Cake$ induces an envy-free partial allocation $\{ I_i = [x_{i-1}, x_i] \}_{i=1}^n$.
\end{restatable} 

This theorem asserts that here the partial allocation $\mathcal{I}=\{I_1, \ldots, I_n\}$ (with $I_i =  [x_{i-1}, x_i]$) satisfies $v_i(I_i) \geq v_i(I_j)$, for all $i, j \in [n]$, and  at most a $\delta$-length piece of the cake (specifically, $[x_n, 1]$) remains unallocated in $\mathcal{I}$. 

\begin{proof}
We will show that if the points $x_0 = 0 \leq x_1 \leq x_2 \leq \ldots \leq x_n \leq 1$ form a $\delta$-ripple division, then the partial allocation $\{I_i  = [x_{i-1}, x_i]\}_{i=1}^n$ is envy-free. Here, the definition of a $\delta$-ripple division (Definition~\ref{defn:deltaRD}) ensures that, for each agent $i \in [n-1]$, the values the two consecutive intervals $I_i$ and $I_{i+1}$ are equal and positive 
\begin{align} \label{equalityRD}
v_i(I_i) = v_i(x_{i-1}, x_i) & = v_i(x_i, x_{i+1}) = v_i(I_{i+1}) >0 
\end{align}
	
Recall that the agents are indexed following the MLRP order, i.e., for each $i \in [n-1]$, the likelihood ratio $f_{i+1}/f_i$ is nondecreasing. We fix an agent $i \in [n]$ and establish envy-freeness with respect to $i$ by considering two complementary cases (i) for agents to the left of $i$, we prove that $v_i(I_1) \leq v_i(I_2) \leq \ldots \leq v_i(I_{i}) $ and (ii) for agents to the right of $i$, we prove that $v_i(I_i) \geq v_i(I_{i+1}) \geq \ldots \geq v_i(I_n)$. \\

\noindent 
Case {(i)}:  Consider any agent $k \in [n]$ such that $k<i$. Given that $f_k$ and $f_i$ bear MLRP (i.e., $f_i/f_k$ is non-decreasing), property (i) of Lemma~\ref{lemma:R2R}, with $a = x_{k-1}$, $b=c=x_k$, and $d=x_{k+1}$, gives us $\frac{v_i(x_{k-1}, x_{k})}{v_k(x_{k-1},x_{k})}  \leq  \frac{v_i(x_{k}, x_{k+1})}{v_k(x_{k},x_{k+1})}$. That is, for the intervals $I_k = [x_{k-1},x_{k}]$ and $I_{k+1} = [x_{k}, x_{k+1}]$ we have
\begin{align} \label{noleftenvy}
\frac{v_i(I_k)}{v_k(I_k)} & \leq  \frac{v_i(I_{k+1})}{v_k(I_{k+1})}
\end{align}
Instantiating equation~(\ref{equalityRD}) for agent $k$, we can simplify inequality (\ref{noleftenvy}) to $v_i(I_k) \leq v_i(I_{k+1})$. Combining this inequality across all $k<i$, we obtain the desired chain of inequalities for agent $i$, i.e., $v_i(I_1) \leq v_i(I_2) \leq \ldots \leq v_i(I_{i}) $.
	
\noindent 
{Case (ii):} Consider any agent $j \in [n]$ such that $j>i$ Given that $f_i$ and $f_j$ bear MLRP (i.e., $f_j/f_i$ is non-decreasing), property (i) of Lemma~\ref{lemma:R2R}, with $a = x_{j-1}$, $b=c=x_j$, and $d=x_{j+1}$,  gives us $\frac{v_j (x_{j-1}, x_j)}{v_i(x_{j-1}, x_j)}  \leq  \frac{v_j(x_{j}, x_{j+1})}{v_i(x_{j},x_{j+1})}$. That is, for the intervals $I_j = [x_{j-1},x_{j}]$ and $I_{j+1} = [x_{j}, x_{j+1}]$ we have 
\begin{align} \label{norightenvy}
\frac{v_j(I_j)}{v_i(I_j)}  \leq  \frac{v_j(I_{j+1})}{v_i(I_{j+1})}  
\end{align}
Instantiating equation~(\ref{equalityRD}) for agent $j$, we can simplify inequality (\ref{norightenvy}) to $v_i(I_j) \geq v_i(I_{j+1})$. Combining this inequality across all $j>i$, we obtain the desired chain of inequalities for agent $i$, i.e., $v_i(I_i) \geq v_i(I_{i+1}) \geq \ldots \geq v_i(I_n)$.

The above two cases establish that agent $i \in [n]$ does not envy any other other agent, i.e., $\mathcal{I}= \{I_1, \ldots, I_n\}$ is an envy-free partial allocation. Indeed, if $\delta = 0$, then $\mathcal{I}$ covers the entire cake, i.e., we obtain an envy-free allocation. 
\end{proof}

Notably, Lemma~\ref{RDexistence} and Theorem~\ref{theorem:RD-EF} (with $\delta=0$) provide a stand-alone proof of existence of envy-free allocations in cake-division instances with MLRP. The next section establishes an algorithmic counterpart of this existential result; specifically, we show that using Lemma~\ref{RDcomputation} one can directly obtain an efficient algorithm for finding envy-free allocations, under MLRP. 

\subsection{Proof of Theorem~\ref{theorem:EFdivision}} \label{proofthm1}

This section restates and proves our main result (Theorem~\ref{theorem:EFdivision}) for envy-freeness.

\EFdivision*

\begin{proof}
Given a cake-division instance $\Cake$, with MLRP, and precision parameter $\eta >0$, we invoke Lemma~\ref{RDcomputation} to find an $\left(\frac{\eta}{\lambda}\right)$-ripple division in $\mathcal{O}\left( {\rm poly} ( n, \log \lambda, \log \frac{1}{\eta} ) \right)$ time. 

Write  $x_0=0 \leq x_1 \leq \dots \leq x_{n-1} \leq x_n \leq 1$ to denote the computed $\left(\frac{\eta}{\lambda}\right)$-ripple division and let $\mathcal{I} = \{I_1, \ldots, I_n\}$ be the corresponding partial allocation; here $I_i = [x_{i-1}, x_i]$. Theorem~\ref{theorem:RD-EF} ensures that $\mathcal{I}$ is envy-free. 

We will show that coalescing the unassigned (in $\mathcal{I}$) piece $[x_n, 1]$ to agent $n$ provides a complete allocation that satisfies envy-freeness, up to $\eta$ precision. Write $\mathcal{I}^* \coloneqq \{I^*_1, I^*_2, \ldots, I^*_{n-1}, I^*_n\}$ to denote this allocation in which the $n$th agent receives the interval $I^*_n \coloneqq [x_{n-1}, 1]$ (equivalently, $I^*_n = I_n \cup [x_n, 1]$) and $I^*_i = I_i$ for the remaining agents $i \in [n-1]$.   

Note that against all agents $j \in [n-1]$, envy-freeness of $\mathcal{I}^*$ directly follows from the fact that the partial allocation $\mathcal{I}$ is envy-free: $v_i(I^*_i) \geq v_i(I_i) \geq v_i(I_j) = v_i(I^*_j)$, for all $i \in [n]$ and all $j \in [n-1]$. 

Finally, we address envy against agent $n$. Recall that $x_i$s form a $\left(\frac{\eta}{\lambda}\right)$-ripple division, hence $x_n \geq 1 - \frac{\eta}{\lambda}$. In addition, the fact that eval queries are $\lambda$-Lipschitz gives us $v_i([x_n, 1]) \leq \eta$ for all agents $i \in [n]$. Hence, for all $i \in [n]$ we have 
\begin{align*}
v_i(I^*_i) & \geq v_i(I_i) \geq v_i (I_n) \tag{since $\mathcal{I}$ is an envy-free partial allocation} \\
& = v_i (I^*_n) - v_i ( [x_n, 1]) \tag{since $I^*_n = I_n \cup [x_n, 1]$} \\
& \geq v_i(I^*_n) - \eta
\end{align*}
Therefore, $\mathcal{I}^*$ satisfies envy-freeness, up to $\eta$ precision: $v_i(I^*_i) \geq v_i(I^*_j) - \eta$ for all $i, j \in [n]$. 

The time complexity obtained via Lemma~\ref{RDcomputation} implies that the parameter $\eta$ can be driven exponentially close to zero, in time that is polynomial in $\log \frac{1}{\eta}$ (i.e., in the bit complexity of $\eta$). Hence, we can find an envy-free allocation, up to arbitrary precision, in $\mathcal{O}\left( {\rm poly} ( n, \log \lambda ) \right)$ time. Theorem~\ref{theorem:EFdivision} now stands proved.
\end{proof}

\section{Pareto Optimality}

This section shows that, with MLRP in hand, one does not loose out on Pareto optimality by imposing the contiguity requirement. That is, under MLRP,   there always exist allocations (i.e., cake divisions with connected pieces) that are Pareto optimal among all cake divisions, with or without connected pieces. Moreover, such allocations conform to the MLRP order. 

This structural result implies that for maximizing welfare we can restrict attention to allocations wherein the intervals are assigned (left to right on the cake) in accordance with the MLRP order. Intuitively, this leads us to a welfare-maximizing algorithm---specifically, a dynamic program---that recursively finds optimal allocations for intervals placed at the left end of the cake (i.e., for intervals of the form $[0,x]$); see Section~\ref{section:social welfare} and Appendix~\ref{section:nsw} for details. 

Subsequently, Section~\ref{section:weller} (Theorem~\ref{theorem:ef-po}) establishes a strong connection between fairness and (Pareto) efficiency in the MLRP context: if the value densities bear MLRP, then \emph{every} envy-free allocation is necessarily Pareto optimal.




  
 

\begin{restatable}{lemma}{MLRPorder}
	\label{theorem:POorder}
	Let $\Cake =\langle [n], (f_i)_i \rangle$ be a cake-division instance in which the value densities satisfy the monotone likelihood ratio property. Then, for every cake division $\mathcal{D}=\{D_1, \ldots, D_n\}$ in $\Cake$ there exists an allocation $\mathcal{J} = \{J_1, \ldots, J_n \}$ such that $v_i(J_i) \geq v_i(D_i)$, for $1 \leq i \leq n$. 
	
	Furthermore, for every Pareto optimal allocation $\mathcal{I} = \{I_1, \ldots, I_n \}$ in $\Cake$, there exists an allocation $\mathcal{I'} = \{I'_1, \ldots, I'_n \}$ with $v_i(I'_i) = v_i(I_i)$ that conforms to the MLRP order, i.e., if $ \ f_{i+1}/f_i $ is nondecreasing in $[0,1]$, then the interval assigned to agent $i$ (i.e., $I'_i$) appears to the left of the interval assigned to the agent $i+1$ (i.e., $I'_{i+1}$). 		
\end{restatable} 
\begin{proof}
Consider a cake division $\mathcal{D}=\{D_1, \ldots, D_n\}$ wherein two consecutive intervals are assigned violating the MLRP order: say, interval $[p,q]$ is assigned to agent $j$ (i.e., this interval is contained in the bundle $D_j$), the adjacent interval $[q,r]$ is assigned to agent $i$, and agent $i$ appears before $j$ in the MLRP order ($f_{j}/f_i$ is non-decreasing over $[0,1]$). 

We will show that in such a case there always exists a point $q' \in [p, r]$ such that $v_i(p,q') \geq v_i(q, r)$ and $v_j(q', r) \geq v_j(p,q)$. That is, one can swap the allocation order between $i$ and $j$ (in the interval $[p,q] \cup [q, r]$) without decreasing the agents' values. Moreover, we note that, if $f_j/f_i$ is strictly increasing in the interval $[p,q] \cup [q, r]$, then this update leads to a strict increase in agent $i$'s or agent $j$'s value. 

Hence, starting with any cake division $\mathcal{D}=\{D_1, \ldots, D_n\}$, we can repeatedly apply the above-mentioned resolution towards the MLRP order and obtain an \emph{allocation} $\mathcal{J} = \{J_1, \ldots, J_n \}$ with the desired property, $v_i(J_i) \geq v_i(D_i)$ for all $i \in [n]$. 

Note that this resolution process also establishes the second part of the theorem, i.e., for every Pareto optimal allocation $\mathcal{I}$, there exists an allocation $\mathcal{I'}$ with $v_i(I'_i) = v_i(I_i)$ that conforms to the MLRP order. 

The remainder of the proof addresses the desired point $q' \in [p,r]$. In particular, we will identify $q'$ such that assigning interval $[p,q']$ to agent $i$ (instead of $[q,r]$) and assigning $[q',r]$ to agent $j$ (instead of $[p,q]$) leads to an increment in values.  

Write $\beta_i \coloneqq \frac{v_i(q,r)}{v_i(p,r)}$ to denote the normalized value of agent $i$ under the initial assignment. Define $q' \in [p,r]$ to be the point that satisfies 
\begin{align}
\frac{v_j(p,q')}{v_j(p,r)} = \beta_i \label{eq:defn-q-prime}
\end{align}


Since the value densities satisfy MLRP, property (ii) of Lemma~\ref{lemma:R2R}, applied to the interval $[p,r]$ and $q' \in [p,r]$, gives us $\frac{v_i(q',r)}{v_i(p,r)}  \leq \frac{v_j(q',r)}{v_j(p,r)} $. Simplifying further we obtain $1- \frac{v_j(q',r)}{v_j(p,r)} \leq 1- \frac{v_i(q',r)}{v_i(p,r)}$, i.e., 
 


\begin{align}
\frac{v_j(p,q')}{v_j(p,r)} \leq \frac{v_i(p,q')}{v_i(p,r)}  \label{ineq:norm-2}
\end{align}


Therefore, we obtain a value bound for agent $i$  
\begin{align*}
\frac{v_i(p,q')}{v_i(p,r)} \geq \frac{v_j(p,q')}{v_j(p,r)} & = \beta_i = \frac{v_i(q, r)}{v_i(p,r)} \tag{via equations (\ref{ineq:norm-2}), (\ref{eq:defn-q-prime}) and the definition of $\beta_i$}
\end{align*} 

That is, agent $i$'s value is preserved through the reassignment, $v_i(p,q') \geq v_i(q,r)$. 

For agent $j$, via property (ii) of Lemma~\ref{lemma:R2R}, with interval $[p,r]$ and $q \in [p,r]$, we have
\begin{align*}
\frac{v_j(q,r)}{v_j(p,r)} \geq \frac{v_i(q,r)}{v_i(p,r)} & = \beta_i  =  \frac{v_j(p,q')}{v_j(p,r)} \tag{by the defintion of $\beta_i$ and equation (\ref{eq:defn-q-prime})} 
\end{align*}

This inequality reduces to $1 - \frac{v_j(p,q')}{v_j(p, r)} \geq 1 - \frac{v_j(q,r)}{v_j(p, r)}$. Simplifying we obtain $\frac{v_j(p, r) - v_j(p,q')}{v_j(p, r)} \geq \frac{v_j(p, r) - v_j(q,r)}{v_j(p, r)}$. Therefore, we have the desired inequality  $v_j(q',r) \geq v_j(p,q)$ and the stated claims follow. 
\end{proof}

\begin{remark}
For cake-division instances with MLRP, we can prove that any allocation $\mathcal{K}$ that conforms to the MLRP order is Pareto optimal (over the set of all cake divisions). Write $0=k_0<k_1< \dots <k_n=1$ to denote the cut-points of $\mathcal{K}$. For contradiction, we assume that $\mathcal{K}$ is not Pareto optimal. That is, there exists a cake-division $\mathcal{L}$ that dominates $\mathcal{K}$ and is Pareto optimal. By Lemma~\ref{theorem:POorder}, we know there exists another Pareto optimal allocation $\mathcal{M}$ with $v_i(M_i)=v_i(L_i)$ for all $i\in[n]$ that conforms to the MLRP order. Write $0=m_0<m_1< \dots <m_n=1$ to denote the cut-points of $\mathcal{M}$. Since, $\mathcal{M}$ Pareto dominates $\mathcal{K}$, we will have $k_i \leq m_i$ for all $i\in[n]$ with at least one strict inequality. This contradicts the fact that $k_n=m_n=1$. Therefore, $\mathcal{K}$ is Pareto optimal.
\end{remark}

\subsection{Extending Weller's Theorem} 
\label{section:weller}

Weller's theorem~\cite{weller1985fair} is a notable result in the cake-cutting literature and it asserts that there always exists some cake division---though, not necessarily with connected pieces---which is both envy-free (fair) and Pareto optimal. While this theorem holds in general, it does not guarantee that envy-freeness and Pareto optimality can be achieved together with allocations. Indeed, there are cake-division instances wherein none of the of envy-free {allocations} are Pareto optimal. 

We show that, by contrast under MLRP, {every} envy-free allocation is Pareto optimal, among all cake divisions (Theorem~\ref{theorem:ef-po}). Therefore, given a cake-division instance with MLRP, the allocation computed by our algorithm (Algorithm \ref{alg:BinSearch}) is not only envy-free but also Pareto optimal, up to an arbitrary precision.

\EFPO*


\begin{proof}
	Write $\mathcal{I} = \{I_1, I_2, \dots, I_n\}$ to denote an envy-free allocation in $\Cake$; here interval $I_i$ is assigned to agent $i \in [n]$. We assume, towards a contradiction, that there exists a cake division $\mathcal{D} = \{D_1, \ldots, D_n\}$ that Pareto dominates $\mathcal{I}$. 
Lemma~\ref{theorem:POorder} implies that in such a case there exists an allocation $\mathcal{J} = \{J_1, \ldots, J_n\}$ which also Pareto dominates $\mathcal{I}$. That is, we have $v_i(J_i) \geq v_i(D_i) \geq v_i(I_i)$, for all agents $i \in [n]$, and there exists some agent $k \in [n]$ such that $v_k(J_k) \geq v_k(D_k) >v_k(I_k)$.

Recall that for an allocation the endpoints of all the constituent intervals are referred to as its {cut points}. We break our analysis into the following two cases depending on whether $\mathcal{I}$ and $\mathcal{J}$ have the same set of cut points. 

\noindent
$\emph{Case 1:}$ The cut points of the allocations $\mathcal{I}$ and $\mathcal{J}$ are identical. In this case, there must exist a permutation $\sigma: [n] \mapsto [n]$ such that $I_{\sigma(i)} = J_i$ for all $i \in [n]$. Since $\I$ is envy-free, we have $ v_i(I_i) \geq v_i(I_{\sigma(i)}) = v_i(J_i)$ for all agents $i \in [n]$. However, this contradicts the fact that $\mathcal{J}$ Pareto dominates the allocation $\I$. 

\noindent
    $\emph{Case 2:}$ The cut points of $\I$ and $\mathcal{J}$ are not identical. Since both the allocations form a partition of the same cake $[0,1]$, there must exist some $s, t  \in [n]$ such that the interval $J_s$ is a strict subset of the interval $I_t$, i.e., $J_s \subset I_t$.  Envy-freeness of $\I$ gives us $v_s(I_s) \geq v_s(I_t) > v_s(J_s)$.  The last strict inequality follows from the fact that the value density $f_s$ of agent $s$ has full support over $[0,1]$ and $J_s$ is a strict subset of $I_t$. This bound $v_s(I_s) > v_s(J_s)$  contradicts the fact that $\mathcal{J}$ Pareto dominates $\I$ and completes the proof. 
\end{proof}



\section{Social Welfare}
\label{section:social welfare}

This section develops an algorithm for social welfare maximization. 
Recall that, under MLRP, Pareto optimal allocations conform to the MLRP order (Lemma \ref{theorem:POorder}). Hence, for maximizing social welfare we can restrict attention to allocations wherein the intervals are assigned (left-to-right on the cake) in accordance with the MLRP order. This observation, in and of itself, leads to a fully-polynomial time approximation scheme for the maximizing social welfare: we can partition the cake into ${\rm poly} (n, 1/\varepsilon)$ contiguous intervals, each of value at most $\varepsilon$, and then solve the problem using a dynamic program. We show that instead of considering a general partition we can identify a set $P$---of $\mathcal{O}(n^2)$ points---such that the cut points of an optimal allocation are contained in $P$. This will enable us to execute a dynamic program focusing only on the points in $P$ and establish Theorem~\ref{theorem:SocialWelfare}.  In sharp contrast to the FPTAS described above, our dynamic program finds an allocation with social welfare $\eta>0$ close to the optimal in time that is dependent on $\log(1/\eta)$.

For value densities $f_i$ and $f_j$, write $L_{ij}$ to denote the set of points at which $f_j$ is at least as large as $f_i$; specifically, $L_{ij} \coloneqq \{x \in [0,1]: {f_j(x)} \geq {f_i(x)} \}$. Since that the densities $f_i$ and $f_j$ are normalized, there must exist a point $x \in [0,1]$ such that $f_j(x) \geq f_i(x)$; in other words, $L_{ij} \neq \emptyset$. Also, observe that this set is bounded below by $0$. Therefore, the greatest lower bound property of $\mathbb{R}$ implies that $L_{ij}$ admits an infimum. We will refer to this infimum as the \emph{switching point} between the two value densities, $p_{ij} \coloneqq \inf_{x \in L_{ij}} \ x$.\footnote{As noted, $p_{ij}$ exists and satisfies $p_{ij} \in [0,1]$. However, since that the value densities $f_i$s are not necessarily continuous (they are only assumed to be Riemann integrable) a point $x$ with the property that $f_i(x) = f_j(x)$ might not exist.} For a cake-division instance $\Cake=\langle[n], \{f_i\}_i \rangle$, we will write $P$ to denote the collection of all switching points $p_{ij}$s, with $ 1\leq i<j \leq n$ (along with $0$ and $1$), i.e., $P\coloneqq \{ p_{i,j} \in [0,1] : 1 \leq i < j \leq n \} \cup \{0, 1\}$; we include the endpoints for cake in $P$ for notational convenience. Also, recall the convention that the agents are indexed following the MLRP order. 

The next lemma provides a useful property about the switching points in the MLRP context, that is reminiscent of the `single-crossing' type condition used in social choice; see \cite{athey2001single, gans1996majority, elkind2014characterization}.
\begin{lemma} \label{switch}
Let $f_i$ and $f_j$ be two value-density functions that satisfy MLRP, i.e., ${f_j}/{f_i}$ is non-decreasing over $[0,1]$. Then,   
    	\begin{itemize}
    		\item[(a)] For all  $y \in [0, p_{i,j})$, the likelihood ratio satisfies ${f_j(y)}/{f_i(y)} <1$.
    		\item[(b)] For all $z \in (p_{ij}, 1]$, the likelihood ratio satisfies ${f_j(z)}/{f_i(z)} \geq 1$.
    	\end{itemize}
	Here, $p_{ij}$ is the switching point between $f_i$ and $f_j$. 
      \end{lemma}

 \begin{proof}
As observed previously, $p_{ij} \in [0,1]$ exists and is unique. We begin by proving part (a) of the stated claim. Consider any point $y \in [0, p_{i,j})$ and assume, towards a contradiction, that the likelihood ratio satisfies $\frac{f_j(y)}{f_i(y)} \geq 1$. This implies that the point $y$ belongs to the set $L_{ij} =\{x \in [0,1]: {f_j(x)} \geq {f_i(x)} \}$. Since $y<p_{ij}$, we get a contradiction to the fact that $p_{ij}$ is the infimum of the set $L_{ij}$.
 	
For proving part (b), consider any point $z \in (p_{ij}, 1]$. Assume, towards a contradiction, that at $z$ we have $\frac{f_j(z)}{f_i(z)} < 1$. Since the likelihood ratio $\frac{f_j(x)}{f_i(x)}$ is non-decreasing over $[0,1]$ (by definition of MLRP), we have $\frac{f_j(t)}{f_i(t)} < 1$ for all $0 \leq t \leq z$. That is, there does not exist a point $t \in [0, z]$ with the property that $f_j(t) \geq f_i(t)$. Hence, $z$ constitutes a lower bound for the set $L_{ij}=\{x \in [0,1]: {f_j(x)} \geq {f_i(x)} \}$. Since $p_{ij} <z$, we get a contradiction to the fact that $p_{ij}$ is the infimum (greatest lower bound) of the set $L_{ij}$. This completes the proof.
\end{proof} 
 
The following corollary asserts that, up to an arbitrary precision, each switching point can be determined efficiently.
 
 \begin{corollary} \label{infirmumswitch}
Let $f_i$ and $f_j$ be two value densities that bear MLRP and let parameter $\gamma \in (0,1)$. Then, in the Robertson-Webb model, we can find an interval of length $\gamma$ that contains  the switching point $p_{ij}$ in $\mathcal{O} \left( \log \left( 1/ \gamma \right) \right)$ time. 
 \end{corollary}
\begin{proof}
Consider intervals of the form $B_k = \left[ (k-1)  \frac{\gamma}{2}, \ k  \frac{\gamma}{2} \right]$ for $k \in \{1,2, \ldots, \frac{2}{\gamma}\}$, i.e., for analysis, we discretize the cake $[0,1]$ evenly into intervals each of length $\gamma/2$. Write $k^*$ to denote the index with the property that $p_{ij} \in B_{k^*} = \left[ (k^*-1)  \frac{\gamma}{2}, \ k^*  \frac{\gamma}{2} \right]$. Part (a) of Lemma~\ref{switch} implies that $v_j(B_k) < v_i(B_k)$ for all $k < k^*$. Similarly, part (b) of Lemma~\ref{switch} gives us $v_j(B_k) \geq v_i(B_k)$ for all $k >k^*$. 

Therefore, applying binary search, we can, in $\mathcal{O} \left( \log \left( 1/ \gamma \right) \right)$ iterations, find the smallest index $k' \in  \{1,2, \dots, \frac{2}{\gamma}\}$ such that $v_j(B_{k'}) < v_i(B_{k'})$ and $v_j(B_{k'+1}) \geq v_i(B_{k'+1})$. Note that the computed index $k'$ satisfies $k' \in \{k^*-1, k^* \}$: if, for contradiction, we have $k' \leq k^*-2$, then it must be the case that $v_j(B_{k'+1}) < v_i(B_{k'+1})$. This inequality contradicts the selection criterion of $k'$. Also, the inequality $k' \geq k^* +1$ would lead to the contradiction $v_j(B_{k'}) \geq v_i(B_{k'})$.

The value comparisons required to execute the binary search can be performed using eval queries. Hence, in $\mathcal{O}\left( \log (1/\gamma) \right)$ iterations we can find an interval $B_{k'} \cup B_{k'+1}$ of length $\gamma$ that contains $p_{ij}$.
\end{proof}

This corollary implies that we can efficiently compute the set of switching points $P =  \{ p_{i,j} \in [0,1] : 1 \leq i < j \leq  n \}$, up to an arbitrary precision. Next, we will establish the usefulness of $P$.

  \begin{lemma} \label{switchingset}
Let $\Cake$ be a cake-division instance in which the value densities satisfy the monotone likelihood ratio property. Then, in $\Cake$, there exists a social welfare maximizing allocation all of whose cut points belong to the set of switching points $P$. 
\end{lemma}
\begin{proof}
Among all allocations that maximize social welfare, consider the ones that conform to the MLRP order; Lemma~\ref{theorem:POorder} ensures that this collection is nonempty. Furthermore, among these optimal allocations select one $\mathcal{S} = \{S_1, S_2, \ldots, S_n\}$ that minimizes $\left| \{s_0=0, s_1, s_2, \ldots, s_n=1 \} \setminus P \right|$; here $s_i$s denote the cut points of $\mathcal{S}$. That is, $\mathcal{S}$ is a social welfare maximizing allocation that uses as many points from $P$ as possible. We will show that $\{s_i \}_i \setminus P = \emptyset$ and, hence, the claim follows. 

Towards a contradiction, assume that there exists a cut point $s_t$ of the allocation $\mathcal{S}$ that does not belong to $P$. Let $S_i =[s, s_t]$ and $S_j=[s_t, s']$ be the two \emph{nonempty} intervals in $\mathcal{S}$ that are separated by $s_t$. Interval $S_i$ is to the (immediate) left of $S_j$ and, since the allocations conform to the MLRP order, we that $i <j$. 

Given that $s_t \notin P$, we know that $s_t \neq p_{ij}$; here $p_{ij}$ is the switching point between $f_i$ and $f_j$. We will show that in this case we can always move $s_t$ towards $p_{ij}$ and obtain another social welfare maximizing allocation that uses more cut points from $P$ than $\mathcal{S}$. This contradicts the choice of $\mathcal{S}$ and establishes the stated claim. Towards this goal, consider two complementary cases \\ 

\noindent
\emph{Case (i):} $s_t < p_{ij}$. In this case we can move $s_t$ to the right without decreasing the social welfare. In particular, if $p_{ij} \in S_i \cup S_j = [s, s']$, then, instead of $S_i$ and $S_j$, we can assign intervals $[s, p_{ij}]$ and $[p_{ij}, s']$ to agents $i$ and $j$, respectively. Since ${f_j(x)} < {f_i(x)}$ for all $x \in [s_t, p_{ij}]$ (Lemma~\ref{switch}, part (a)), such an update increases the social welfare. This contradicts the optimality (with respect to social welfare) of $\mathcal{S}$. A similar argument holds if $p_{ij} > s'$. Here, we can assign $[s, s']$ entirely to agent $i$ (and an empty set to agent $j$).  For all $x \leq s' < p_{ij}$, we have ${f_j(x)} < {f_i(x)}$ (Lemma~\ref{switch}, part (a)). Therefore, the reassignment increase the social welfare and leads to a contradiction. \\

\noindent
\emph{Case (ii):} $s_t > p_{ij}$. In this case we can move $s_t$ to the left (towards $p_{ij}$). If we have $p_{ij} \in S_i \cup S_j$, then assigning intervals $[s, p_{ij}]$ and $[p_{ij}, s']$ to agents agents $i$ and $j$, respectively, does not decrease the social welfare (Lemma~\ref{switch}, part (b)). Though, at the same time, it does provide an allocation that uses more cut points from $P$ and, hence, contradicts the choice of $\mathcal{S}$. On the other hand, if  $p_{ij} <s$, then we can assign the entire interval $[s,s']$ to agent $j$. As before, the reassignment does not decrease the social welfare. However, it does decrease the cardinality of the set difference between the cut points and $P$. This contradicts the selection criterion of $\mathcal{S}$.

Hence, the cut points of $\mathcal{S}$ satisfy $\{s_i\}_i \subseteq P$ and the stated claim follows. 
\end{proof}

We now present the main result of this section.

\SocialWelfare*

\begin{proof}
 Given a cake-division instance $\Cake$ with MLRP, write $\mathcal{S}^* = \{S^*_1, S^*_2, \dots, S^*_n\}$ to denote the allocation identified in Lemma~\ref{switchingset}; in particular,  $\mathcal{S}^*$ is a social welfare maximizing allocation whose cut points $\{0=s^*_0, s^*_1, \dots, s^*_n=1\}$ belong to the set of switching points $P$.  
 
For a precision parameter $\eta >0$ and for each switching point $p_{ij} \in P$, we invoke Corollary~\ref{infirmumswitch} to find $\widehat{p}_{ij} \in [0,1]$ with the property that $|\widehat{p}_{ij}-p_{ij}| \leq \frac{\eta}{n \lambda}$. Write $\widehat{P}$ to denote the set of these estimates, $\widehat{P} \coloneqq \{ \widehat{p}_{ij} : 1 \leq i < j \leq n \} \cup \{0,1\}$.\footnote{As in the case of $P$, we include $0$ and $1$ in $\widehat{P}$ for ease of presentation.} Applying Corollary~\ref{infirmumswitch} to each $\widehat{p}_{ij}$, we get that the set $\widehat{P}$ can be computed in $\mathcal{O}({\rm poly}(n, \log \lambda, \log \frac{1}{\eta}))$ time. 

Next we will show that there exists an allocation $\widehat{\mathcal{S}}$ whose cut points are contained in $\widehat{P}$ and this allocation has near-optimal social welfare, $\SW(\widehat{\mathcal{S}}) \geq \SW(\mathcal{S}^*) - \eta$. Specifically, for each cut point $s^*_i$ of $\mathcal{S}^*$, let $\widehat{s}_i$ denote its closest point in $\widehat{P}$, i.e., $\widehat{s}_i \coloneqq \argmin_{ \widehat{p} \in \widehat{P}} | s^*_i - \widehat{p} |$; here we break ties, say, lexicographically. By construction of $\widehat{P}$ we have $|s^*_i - \widehat{s}_i | \leq \frac{\eta}{n \lambda}$, for each $0 \leq i \leq n$.

For the optical allocation $\mathcal{S}^* = \{S^*_1, \ldots, S^*_n \}$ we have $S^*_i = [s^*_{i-1}, s^*_i]$ for all $i \in [n]$. Write allocation $\widehat{\mathcal{S}} \coloneqq \{\widehat{S}_1, \widehat{S}_2, \ldots, \widehat{S}_n\}$, where interval $\widehat{S}_i = [\widehat{s}_{i-1}, \widehat{s}_i]$ is assigned to agent $i \in [n]$. The cut points of allocations $\widehat{\mathcal{S}}$ are contained in $\widehat{P}$. Also, given that $\mathcal{S}^*$ conforms to the MLRP order, so does $\widehat{\mathcal{S}}$. 
 
Using the bounds $|\widehat{s}_i - s^*_i| \leq \frac{\eta}{n \lambda}$, for $0 \leq i \leq n$, and the fact that eval queries are $\lambda$-Lipschitz, we obtain $|v_i(\widehat{S}_i) - v_i(S^*_i)| \leq \frac{\eta}{n}$, for all agents $i \in [n]$. Summing over all $i\in [n]$, we obtain the desired social-welfare bound: $|\sum_{i \in [n]} v_i(\widehat{S}_i) - \sum_{i \in [n]}v_i(S^*_i)| = |\SW(\widehat{\mathcal{S}}) - \SW(\mathcal{S}^*)|  \leq \eta$.

Therefore, there exists an allocation $\widehat{\mathcal{S}}$ with the properties that (i) $\widehat{\mathcal{S}}$ has near-optimal social welfare, (ii) cut points of $\widehat{\mathcal{S}}$ are contained in $\widehat{P}$, and (iii) $\widehat{\mathcal{S}}$ conforms to the MLRP order. To complete the proof of the theorem we will show that, among all allocations that satisfy properties (ii) and (iii), we can find one that maximizes social welfare. We accomplish this algorithmic result by a simple dynamic program. 

Recall that cardinality of the set $\widehat{P}$ is $\mathcal{O}(n^2)$ and this set can be computed in $\mathcal{O}({\rm poly}(n, \log \lambda, \log \frac{1}{\eta}))$ time using eval queries. We index the elements of the computed set $\widehat{P} = \{ \widehat{p}_t \}_t$ such that $0 = \widehat{p}_0  < \widehat{p}_1 < \ldots < \widehat{p}_{|\widehat{P}|} = 1$  

For each $k \in [n]$ and $1 \leq t \leq |\widehat{P}|$, we write $M(k, t)$ to denote the maximum social welfare that one can achieve by allocating the interval $[0, \widehat{p}_t]$ among the first $k$ agents (in order).\footnote{By convention, the agents are indexed following the MLRP order.}

The following recursive equation for $M(k, t)$ gives us the desired dynamic program
\begin{align*}
M(k, t) & \coloneqq \max_{1 \leq t' \leq t } \left\{ M(k-1, t') + v_k(\widehat{p}_{t'}, \widehat{p}_t) \right\}
\end{align*}

Here, we initialize $ M(1,t) \coloneqq v_1(0, \widehat{p}_t) $ for all $ 1 \leq t \leq |\widehat{P}|$. One can directly show, via induction, that $M(n, |\widehat{P}|)$ is equal to the optimal social welfare among allocations that satisfy the above-mentioned properties (ii) and (iii). That is, $M(n, |\widehat{P}|) \geq \SW(\mathcal{S}^*) - \eta$. Hence, the dynamic program gives us the desired allocation. For the runtime analysis, note that the dynamic program runs in $\mathcal{O}( n^2 |\widehat{P}|)$ time and only requires eval queries. 

Overall, we can find an allocation with social welfare $\eta$ close to the optimal in time $\mathcal{O}({\rm poly}(n, \log \lambda, \log \frac{1}{\eta}))$. Since the precision parameter $\eta$ can be driven exponentially close to zero, in time that is polynomial in $\log \frac{1}{\eta}$ (i.e., in the bit complexity of $\eta$) the stated claim follows. 
\end{proof}

\begin{remark2}
For cake-division instances with MLRP, Theorem~\ref{theorem:SocialWelfare} also shows that we can efficiently maximize social welfare among all cake divisions, and not necessarily among allocations. In particular, let $\mathcal{D}^* = \{D^*_1, D^*_2, \dots, D^*_n\}$ denote a cake division that maximizes social welfare in a given instance. We know that there exists an allocation $\mathcal{J}=\{J_1, \ldots, J_n \}$ such that $v_i(J_i) \geq v_i(D_i)$, for all $i \in [n]$ (Theorem~\ref{theorem:POorder}). Summing over $i$ we get $\SW(J) \geq \sum_{i=1}^n v_i(D^*_i)$.  Therefore, the allocation computed through Theorem~\ref{theorem:SocialWelfare} is (near) optimal among all cake divisions. 
\end{remark2}

\section{Egalitarian Welfare} \label{section:max-min}

This section presents an algorithm for maximizing egalitarian welfare in cake-division instances with MLRP (Theorem~\ref{theorem:Maxmin}).



Towards this end, we will define a ``moving-knife'' procedure that leads to the desired algorithm. Given a target value $\tau >0$, we consider cut points obtained by iteratively selecting intervals that are of value $\tau$ to agents $1$ through $n$, respectively. That is, the first cut point $x_1$ is chosen to satisfy $v_1(0,x_1) = \tau$. Inductively, for $i \geq 2$, the cut point $x_i$ is selected such that $v_i(x_{i-1}, x_i ) = \tau$. Recall that the agents are indexed following the MLRP order. To specify these points, we recursively define $n$ functions, $\MK_i: \mathbb{R}_+ \mapsto [0,1]$, for $i \in [n]$: 
\begin{align*}
\MK_1(\tau) &\coloneqq  \Cut_1(0, \tau) \nonumber \\ 
\MK_i(\tau) &\coloneqq  \Cut_i(\MK_{i-1}(\tau), \tau)
\end{align*}

Here, $\MK_n(\tau)$ denotes the last cut point which we obtain by executing the moving-knife procedure with target value of $\tau \in [0,1]$ for every agent. Given that $\MK_i$s can be expressed as a composition of cut queries, these functions can be efficiently computed in the Robertson-Webb query model.

The following proposition establishes that $\MK_n$ is monotonic.  


\begin{proposition} \label{proposition:MKn}
For any cake-division instance $\Cake = \langle [n], \{f_i\}_i \rangle$ and each $i \in [n]$, the function $\MK_i$ is monotonically increasing. 
\end{proposition}

\begin{proof}
Consider two target values $\tau, \tau' \in \mathbb{R}_+$ such that $\tau \leq \tau'$. We will show by inducting over $i \in [n]$ that $\MK_i(\tau) \leq \MK_i(\tau')$. For the base case $i =1$, observe that $\Cut_1(0,\tau) \leq \Cut_1(0, \tau')$, since $f_1(x) \geq 0$ for all $x \in [0,1]$, In other words, we have $\MK_1(\tau) \leq \MK_1(\tau')$. Next, with the induction hypothesis $\MK_{i-1}(\tau) \leq \MK_{i-1}(\tau')$ in hand, we will establish $\MK_{i}(\tau) \leq \MK_{i}(\tau')$:
\begin{align*}
\MK_i(\tau) =  \Cut_i(\MK_{i-1}(\tau), \tau) \leq \Cut_i(\MK_{i-1}(\tau'), \tau) \leq \Cut_i(\MK_{i-1}(\tau'), \tau') = \MK_{i}(\tau')
\end{align*}
Here, we use the fact that the value densities are nonnegative. This establishes the stated claim. 
\end{proof}

We now establish the main result for egalitarian welfare.
\Maxmin*

\begin{proof} 
For the given cake-division instance $\Cake$, write $\mathcal{R}^* = \{R^*_1, R^*_2, \ldots, R^*_n\}$ to denote an allocation that maximizes egalitarian  welfare among all allocations in $\Cake$. Without loss of generality we can assume that $\mathcal{R}^*$ is Pareto optimal an conforms to the MLRP order (Lemma~\ref{theorem:POorder}). Write $\tau^* \coloneqq \RW(\mathcal{R}^*) =  \min_i v_i(R^*_i)$.  

Note that $\MK_n(0) = 0$ and $\MK_n(1) = 1$; by convention, $\Cut_i( \ell, \tau)$ is truncated to $1$ iff $\tau$ is greater than the entire value to the right of $\ell$, i.e., if $v_i(\ell, 1) < \tau$.

We will say that a target value $\tau$ is \emph{feasible} iff $v_i(\MK_{i-1}(\tau),\MK_i(\tau)) = \tau$ for all $i \in [n]$. That is, executing the moving-knife procedure with target value $\tau$ ensures that each agent receives an interval $[\MK_{i-1}(\tau), \MK_i (\tau)]$ of value $\tau$. Otherwise, we say that $\tau$ is \emph{infeasible}. Note that if $\tau$ is infeasible, then for some $i \in [n]$ the $\MK_i(\tau)$ gets truncated at one--in such a case, $i$ and subsequent agents $j >i$ do not receive intervals of value $\tau$. Hence, $\tau=0$ is a feasible target value, whereas $\tau=1$ is infeasible. Note that we can efficiently determine whether a given $\tau$ is feasible or not.\footnote{As mentioned previously, the functions $\MK_i$s can be computed efficiently in the Robertson-Webb model.}

Proposition~\ref{proposition:MKn} implies that we have monotonicity with respect to the feasibility of target values. Specifically, ({\rm P}): for all $\tau \leq \tau'$,  if $\tau$ is infeasible, then so is $\tau'$. To establish property ({\rm P}) we note that, for an infeasible $\tau$, there exists, by definition, an agent $i \in [n]$ such that $v_i(\MK_{i-1}(\tau), \MK_i(\tau)) < \tau$. This happens when $\Cut_i(\MK_{i-1}(\tau), \tau) = \MK_i(\tau)$ gets truncated to $1$. Applying Proposition~\ref{proposition:MKn} to agent $i-1$, with $\tau \leq \tau'$, we obtain $\MK_{i-1}(\tau) \leq \MK_{i-1}(\tau')$. Therefore, for agent $i$ the value of the remaining cake $[\MK_{i-1}(\tau'),1]$ is less than $\tau'$. Once again $\Cut_i(\MK_{i-1}(\tau'), \tau')$ gets truncated at $1$ and we obtain $v_i \left(\MK_{i-1}(\tau'), \MK_{i}(\tau') \right) < \tau'$, i.e., $\tau'$ is infeasible.


With a precision parameter $\eta >0$ in hand, we perform binary search over integer multiples of ${\eta}$, i.e., over the set $\left\{k \eta \right\}_{k = 0}^{1/\eta}$. Recall that  the values of the agents are normalized and, hence, $\tau^* \in [0,1]$. Also, $\tau=0$ is feasible, while $\tau=1$ is infeasible., 

This switch in feasibility (between $0$ and $1$) along with property ({\rm P}), imply that there exists a unique index $k_0 \in \{0, 1, \ldots, 1/\eta \}$ with the property that $k_0 \eta$ is feasible and $(k_0+1) \eta$ is infeasible. We can identify $k_0$ by $\mathcal{O}\left( \log (1/\eta) \right)$ iterations of binary search. 

The relevant observation here is that $\tau^* \in [k_0 \eta, (k_0 + 1) \eta]$. Indeed, the optimal value $\tau^*$ is feasible: $\mathcal{R}^*=\{R^*_1, \ldots, R^*_n\}$ conforms to the MLRP order and $v_i(R^*_i) \geq \tau^*$ for all $i \in [n]$. Hence, property ({\rm P}) ensures that $\tau^*$ cannot be greater than the infeasible target $(k_0 + 1) \eta$. Furthermore, the optimality of $\tau^*$ gives us $ \tau^* \geq k_0 \eta$.

For any feasible $\tau$, the moving-knife procedure provides an allocation with egalitarian welfare at $\tau$. Specifically, for the computed index $k_0$, consider allocation $\widehat{\mathcal{R}} = \{\widehat{R}_1, \widehat{R}_2, \ldots, \widehat{R}_n\}$ in which $\widehat{R}_i = [\MK_{i-1}(k_0 \eta), \MK_i(k_0 \eta) ]$, for $i \in [n-1]$ and $\widehat{R}_n = \left[\MK_{n-1}(k_0 \eta), \MK_{n}(k_0 \eta)  \right] \cup [\MK_{n}(k_0 \eta), 1]$. For notational convince here we set $\MK_0(k_0 \eta) = 0$. Feasibility of $k_0 \eta$ ensures that $v_i(\widehat{R}_i) = k_0 \eta \geq \tau^* - \eta$ for all $i \in [n]$. Therefore, we have $\RW(\widehat{\mathcal{R}}) \geq \tau^* - \eta$. 

For the runtime analysis, note that the binary search finds the desired index $k_0$ in $\mathcal{O}({\rm poly}(n,  \log \frac{1}{\eta}))$ time. The dependency on $\lambda$ stems from the fact that the bit-complexity of the output (i.e., of the computed cut points) can be $\mathcal{O} \left( \log \frac{\eta}{\lambda}\right)$. 

Overall, these arguments show that we can find an allocation $\widehat{\mathcal{R}}$ with egalitarian welfare $\eta$ close to the optimal in time $\mathcal{O}({\rm poly}(n, \log \lambda, \log \frac{1}{\eta}))$. The precision parameter $\eta$ can be driven exponentially close to zero in time that is polynomial in $\log \frac{1}{\eta}$ and, hence, the stated claim follows.
\end{proof}



\section{Nash Social Welfare } \label{section:nsw}
This section presents an FPTAS for maximizing Nash social welfare in cake-division instances with MLRP.

\NashSocialWelfare*

\begin{proof}
	Given a cake-division instance $\Cake$ with MLRP, write $\mathcal{A}^* = \{A^*_1, A^*_2, \dots, A^*_n\}$ to denote an allocation that maximizes the Nash social welfare in $\Cake$. Arunachaleswaran et al. \cite{arunachaleswaran2019fair} have shown that the bundles in any Nash optimal allocation $\mathcal{A}^*= \{A^*_1, A^*_2, \dots, A^*_n\}$ satisfy $v_i(A^*_i) \geq \frac{1}{4} v_i(A^*_j)$ for all $i,j \in[n]$; this result holds even in the absence of MLRP. For a fixed agent $i \in [n]$, we  sum the inequalities $v_i(A^*_i) \geq \frac{1}{4} v_i(A^*_j)$ over all $j \in [n]$ to obtain $v_i(A^*_i) \geq \frac{1}{4n}$. {Recall that $v_i(0,1)=1$ for all $i \in [n]$.} 


For an approximation parameter $\varepsilon >0$, we consider points $c_0= 0 <  c_1 < c_2 < \ldots < c_N=1$ such that  $v_i(c_{t-1}, c_{t}) \leq \frac{\varepsilon}{8n}$ for all $t \in [N]$ and for all agents $i \in [n]$. In particular, we start with $c_0=0$ and iteratively select points $c_t \coloneqq \min_i \ \Cut_i \left(c_{t-1},  \frac{\varepsilon}{8n} \right)$, for all $t \geq 1$. Note that the following upper bound holds for total number of selected points $N \leq \frac{8 n^2}{\varepsilon}$. 

We round each cut point in the optimal allocation $\mathcal{A}^*$ to its closest point in the collection $\{c_i\}_{i=0}^N$. This leads us to another allocation, $\widehat{\mathcal{A}} = \{\widehat{A}_1, \widehat{A}_2, \ldots, \widehat{A}_n\}$, wherein interval $\widehat{A}_i$ is assigned to agent $i \in [n]$. By construction, the cut points of $\widehat{\mathcal{A}}$ are contained in the set $\{c_t\}_{t=0}^N$. Also, since $\mathcal{A}^*$  conforms to the MLRP order (Lemma~\ref{theorem:POorder}), so does $\widehat{\mathcal{A}}$. Next we show that the Nash social welfare of $\widehat{\mathcal{A}}$ is comparable to that of $\mathcal{A}^*$. For all $i \in [n]$, we have 
	\begin{align*}
	v_i(\widehat{A}_i) &\geq v_i(A^*_i) - \frac{\varepsilon}{4n} \tag{since $v_i(c_{t-1}, c_{t}) \leq \frac{\varepsilon}{8n}$ for all $t \in [N]$}\\
	& \geq v_i(A^*_i) - \varepsilon v_i(A^*_i)  \tag{since $v_i(A^*_i) \geq \frac{1}{4n}$ } \\
	& \geq (1 - \varepsilon) \ v_i(A^*_i) 
	\end{align*}
	Multiplying the above inequality over  $i\in [n]$, we obtain $\NSW(\widehat{\mathcal{A}}) = \left( \prod_{i=1}^n v_i(\widehat{A}_i) \right)^{1/n} \geq (1-\varepsilon) \left( \prod_{i=1}^n v_i(A^*_i) \right)^{1/n} = (1 - \varepsilon) \ \NSW(\mathcal{A}^*)$.
	
	Therefore, there exists an allocation $\widehat{\mathcal{A}}$ with the properties that (i) $\widehat{\mathcal{A}}$ has Nash social welfare at least $(1- \varepsilon)$ times the optimal (Nash social welfare), (ii) the cut points of $\widehat{\mathcal{A}}$ are contained in the set $\{ c_t \}_t$, and (iii) $\widehat{\mathcal{A}}$ conforms to the MLRP order. To complete the proof of the theorem we will show that, among all the allocations that satisfy (ii) and (iii), we can find (via a dynamic program) one that maximizes the Nash social welfare. 
	
	For $t \in [N]$ and $k \in [n]$, we write $H(k,t)$ to denote the optimal Nash product (i.e., the product of valuations) that one can achieve by allocating the interval $[0,c_t]$ among the first $k$ agents (in order).
	\begin{align*}
	H(k,t) \coloneqq \max_{1 \leq t' \leq t} \left\{  H(k-1, c_{t'}) \cdot v_k(c_{t'},c_t)  \right\}
	\end{align*}
	Here, we initialize the dynamic program by setting $H(1,t) \coloneqq v_1(0,c_t)$ for all $1 \leq t \leq N$. One can show, via induction, that $H(n, N)$ is the optimal Nash product among allocations that satisfy properties (ii) and (iii). That is, $(H(n, N))^{1/n} \geq (1 - \varepsilon) \ \NSW(\mathcal{A}^*)$. Hence, the dynamic program computes the desired allocation using $\mathcal{O}(n^2 N)$ eval queries.
	
	Therefore, we can find an allocation with Nash social welfare at least $(1- \varepsilon)$ times the optimal in $\mathcal{O} \left( {\rm poly} \left(n, 1/\varepsilon, \log \lambda \right) \right)$ time; the dependency on $\log \lambda$ stems from the fact that the bit-complexity of the output (i.e., of the computed cut points) can be $\mathcal{O} \left( \log \frac{\varepsilon}{\lambda}\right)$.
Overall, we get that  maximizing Nash social welfare admits an FPTAS under MLRP.
\end{proof}

\section{Conclusion and Future Work}
The current work studies algorithmic aspects of contiguous cake division under the monotone likelihood ratio property. The scope of this property ensures that the developed algorithms are applicable in various cake-division settings. We also note that while under MLRP the value densities must have full support, the developed framework is somewhat robust to this requirement. For example, our results extend to the class of non-full-support value densities considered in \cite{alijani2017envy}. In particular, Alijani et al. \cite{alijani2017envy} established that a contiguous envy-free cake division can be efficiently computed if every agent uniformly values a single interval and these intervals satisfy an \emph{ordering} property. Appendix \ref{appendix:structured-perturbations-for-MLRP} shows that here one can modify the value densities to a small degree and obtain MLRP (with full support).  Hence, applying our results, one can efficiently compute an allocation with arbitrary small envy in the modified instance and, hence, also in the original one.  Generalizing such ideas to address, say, value densities that bear first-order stochastic dominance is an interesting direction of future work.

For instances with MLRP, finding an allocation that maximizes various welfare notions among the set of envy-free allocations is an important thread for future work. Another relevant notion of fairness in the context of cake cutting is that of a \emph{perfect division}~\cite{alon1987splitting}. In such a division $\mathcal{D} = \{D_1, D_2, \ldots, D_n\}$, each agent $i \in [n]$ values every piece at $1/n$, i.e., $v_i(D_j)  = 1/n$ for all $i, j \in [n]$. In contrast to the other solution concepts considered in the present paper, perfect divisions are not guaranteed to exist under the contiguity requirement; a perfect \emph{allocation} might not exist even with MLRP (Appendix \ref{appendix:perfect-cuts-nonexample}).  However, Alon~\cite{alon1987splitting} has shown that a perfect division with $n(n-1)$ cuts always exists. Perfect cake divisions are particularly useful since they lead to truthful mechanisms for cake division \cite{mossel2010truthful, chen2013truth}. Hence, developing efficient algorithms to find (noncontiguous) perfect divisions under MLRP is a relevant thread for future work. 

More broadly, it would be interesting to identify tractable classes through MLRP in other computational social choice contexts, such as discrete fair division and voting. 








\section*{Acknowledgements}
We thank Manjunath Krishnapur for helpful discussions and references. Siddharth Barman gratefully acknowledges the support of a Ramanujan Fellowship (SERB - {SB/S2/RJN-128/2015}) and a Pratiksha Trust Young Investigator Award. Nidhi Rathi's research is generously supported by an IBM PhD Fellowship.

\bibliographystyle{alpha}
\bibliography{references}

\appendix

\section{Implications of MLRP}

\subsection{Proof of Lemma~\ref{lemma:R2R}}
\label{appendix:proof-race-to-ratio}
In this section we restate and prove Lemma~\ref{lemma:R2R}. 

\RacetotheRatios*

 \begin{proof}
 Given that $f_i$ and $f_j$ bear MLRP, we will first prove that they satisfy property (i). For  $b \leq c$, we have 
\begin{align} \label{ineq:MLRP}
 	\frac{f_{j}(x)}{f_i(x)} \leq 	\frac{f_{j}(b)}{f_i(b)} \leq 	\frac{f_{j}(y)}{f_i(y)} \quad \text{for all} \ x \in [a,b] \ \ \text{and for all} \ y \in [c,d]
\end{align}

Recall that MLRP value densities are, by definition, positively valued, $f_i(x) > 0$ for all $x \in [0,1]$. Therefore, equation (\ref{ineq:MLRP}) gives us $f_{j}(x) \leq  \ \frac{f_{j}(b)}{f_i(b)}  \ f_i(x)$ for all $x \in [a,b]$. Integrating we obtain $\int \limits_a^b f_{j}(x) dx \leq \int \limits_a^b \frac{f_{j}(b)}{f_i(b)}  \ f_i(x) dx$ and, hence,\footnote{Recall that the integral of a positive function is positive.} 
\begin{align}
\frac{\int_a^b f_{j}(x)dx}{\int_a^b f_i(x)dx} \leq & \ \frac{f_{j}(b)}{f_i(b)} \label{ineq:left-sandwich}
\end{align}

Starting with equation (\ref{ineq:MLRP}) and integrating over the interval $[c,d]$, we can also establish the following equality
 \begin{align}
   \frac{f_{j}(b)}{f_i(b)} \leq  \frac{\int_c^d f_{j}(x)dx}{\int_c^d f_i(x)dx} \label{ineq:right-sandwich}
 \end{align}
 
Equations (\ref{ineq:left-sandwich}) and (\ref{ineq:right-sandwich}) lead to property (i):
 \begin{align*}
  \frac{\int_a^b f_{j}(x)dx}{\int_a^b f_i(x)dx}  \leq  \frac{\int_c^d f_{j}(x)dx}{\int_c^d f_i(x)dx}. 
 \end{align*}
 
 Next we will prove that properties (i) and (ii) are equivalent. Since $f_i$ and $f_j$ satisfy property (i), the equivalence of properties (i) and (ii) will imply that they satisfy property (ii) as well; thereby completing the proof.
 
To establish that property (i) implies property (ii), we instantiate (i) over the intervals $[a,x]$ and $[x,b]$: $\frac{\int_a^x f_{j}}{\int_a^x f_{i}} \leq \frac{\int_x^b f_{j}}{\int_x^b f_{i}}$. Cross multiplying the terms\footnote{Recall that the value densities are strictly positive.} and adding one to both sides of the inequality, gives us $\frac{\int_a^x f_{j}}{\int_x^b f_{j}} + 1 \leq \frac{\int_a^x f_{i}}{\int_x^b f_{i}} + 1$. Simplifying further we obtain  $\frac{\int_a^x f_{j} + \int_x^b f_{j}}{\int_x^b f_{j}} \leq \frac{\int_a^x f_{i} + \int_x^b f_{i}}{\int_x^b f_{i}}$. This gives us the desired bound
\begin{align}
\frac{\int_x^b f_{i}}{\int_a^b f_i} \leq \frac{\int_x^b f_{j}}{\int_a^b f_{j}} \label{ineq:prop-three}
\end{align}

For the reverse direction, i.e., (ii) implies (i), we begin by cross multiplying the terms in inequality (\ref{ineq:prop-three}) and subtracting  one from both sides yield $\frac{\int_a^b f_{j}}{\int_x^b f_{j}} - 1  \leq \frac{\int_a^b f_i }{ \int_x^b f_i} - 1$. This inequality simplifies to $\frac{\int_a^x f_{j}}{\int_x^b f_{j}} \leq \frac{\int_a^x f_i}{\int_x^b f_i}$. That is, we obtain property (i) for intervals $[a,x]$ and $[x,b]$. Reapplying this bound (with $a$, $x$, and $b$ set appropriately) shows that (ii) implies (i).  This completes the proof. 
\end{proof}

\subsection{Efficiently Finding the MLRP Order}
\label{appendix:find-mlrp-order}

This section shows that, if in a cake-division instance the underlying value densities bear MLRP, then we can efficiently find the MLRP order in the Robertson-Webb query model. Specifically, through eval queries $\{\Eval_i(1/2, 1)\}_{i\in [n]}$ we find the value that each agent $i \in [n]$ has for the piece $[1/2,1]$ of the cake, and sort the agents according to the increasing order (with ties broken arbitrarily) of these values. Lemma~\ref{MLRPorder} shows that this sorting order is in fact the MLRP order of the underlying value densities.

To prove this claim, we first state a well-known result (in Lemma~\ref{SD}) that MLRP implies first-order stochastic dominance; we provide a proof here for completeness.
Recall that, given two probability density functions $f_i$ and $f_j$ over $[0,1]$, density $f_j$ is said to have \emph{first-order stochastic dominance} over $f_i$ iff $\int_t^1 f_j(x)dx \geq \int_t^1 f_i(x)dx$, for all $t \in [0,1]$, and there exists at least one $t' \in [0,1]$ such that $\int_{t'}^1 f_j(x)dx > \int_{t'}^1 f_i(x)dx$. 

\begin{lemma} \label{SD}
Let $f_i$ and $f_j$ be two (ordered) value-density functions that satisfy the monotone likelihood ratio property: for every $ 0 \leq x \leq y \leq 1$ we have $\frac{f_j(x)}{f_i(x)} \leq \frac{f_j(y)}{f_i(y)}$. Then, the density $f_j$ has first-order stochastic dominance over $f_i$. 
\end{lemma}

\begin{proof}
Given that $f_i$ and $f_j$ bear MLRP, we consider property (ii) of Lemma~\ref{lemma:R2R}, with $a = 0$, $b=1$, and $x=t$, for any $0 \leq t \leq 1$, to obtain ${\int \limits_t^1 f_i}  \leq \int \limits_t^1 f_j $.
Here, we use the fact that the valuations are normalized, $\int_0^1 f_i = \int_0^1 f_j = 1$. 

Furthermore, since the two densities are distinct, there exists a point $t' \in [0,1]$ such that $\int_{t'}^1 f_j $ is not equal to $\int_{t'}^1 f_i$. For such a point $t'$, a strict inequality must hold, $\int_{t'}^1 f_j  > \int_{t'}^1 f_i$.
\end{proof}

Note that, since both MLRP and first-order stochastic dominance are transitive properties, the above lemma directly extends to a collection of probability density functions.

Next we prove that, given a cake-division instance, with the promise that the underlying value densities satisfy MLRP, one can find the MLRP order efficiently. 

\begin{lemma} \label{MLRPorder}
Let $\Cake$ be a cake-division instance in which the value-density functions satisfy the monotone likelihood ratio property. Then, in the Robertson-Webb query model, we can find the MLRP order of $\Cake$ in polynomial time.  
\end{lemma}

\begin{proof}
Through eval queries $\{\Eval_i(1/2, 1)\}_{i\in [n]}$, we find the value that each agent $i \in [n]$ has for the piece $[1/2,1]$ of the cake, and sort the agents according to the increasing order (with ties broken arbitrarily) of these values. We will show that that this sorting order, say $\pi :[n] \mapsto [n]$, is in fact the MLRP order of the underlying value densities. This directly provides an efficient algorithm for finding the MLRP order. 

Consider two agents $i, j \in [n]$ such that $i$ appears before $j$ in the MLRP order, i.e., the likelihood ratio $f_j(x)/f_i(x)$ is non-decreasing in $x \in [0,1]$. Below, we will address the settings in which the value densities of $i$ and $j$ are distinct. Otherwise, if the value densities are identical, then $\Eval_i(1/2, 1) = \Eval_j(1/2,1)$ and they will be placed appropriately in the sorting order $\pi$. 

Given that $f_i$ and $f_j$ are distinct we get, via Lemma~\ref{SD}, that $f_j$ has first-order stochastic dominance over $f_i$. Therefore, agent $j$'s value for the piece $[1/2,1]$ is at least as large as agent $i$'s value for this piece, $\int_{1/2}^1 f_j \geq \int_{1/2}^1 f_i$, i.e., $\Eval_j(1/2, 1) \geq \Eval_i(1/2, 1)$. We will show that, under MLRP, this inequality is in fact strict and, hence, $\pi$ correctly identifies the MLRP order among each pair of agents $j$ and $i$. 

The first-order stochastic dominance between $f_j$ and $f_i$ also ensures that there exists a point $t' \in [0,1]$ such that 
\begin{align}
\int_{t'}^1 f_j  > \int_{t'}^1 f_i \label{ineq:defn-t-prime}
\end{align}
We consider two complementary cases $(i)$ $t'\leq 1/2$ and $(ii)$ $t'>1/2$. In both of these cases we assume, towards a contradiction, that $\int_{1/2}^1 f_j = \int_{1/2}^1 f_i$.  

\noindent
{\emph Case (i):} $t' \leq 1/2$. Note that in this case equation (\ref{ineq:defn-t-prime}) expands to $\int_{t'}^{1/2} f_j + \int_{1/2}^1 f_j > \int_{t'}^{1/2} f_i + \int_{1/2}^1 f_i$. Since, $\int_{1/2}^1 f_j = \int_{1/2}^1 f_i$, we have $\int_{t'}^{1/2} f_j > \int_{t'}^{1/2} f_i$.

On the other hand, applying Lemma~\ref{lemma:R2R}, property (i), to intervals $[t',1/2]$ and $[1/2,1]$ gives us $ \frac{\int_{t'}^{1/2} f_j}{\int_{t'}^{1/2} f_i} \leq \frac{\int_{1/2}^1 f_j}{\int_{1/2}^1 f_i} = 1$. This inequality leads to the desired contradiction, $\int_{t'}^{1/2} f_j \leq \int_{t'}^{1/2} f_i$. 

\noindent
{\emph Case (ii):} $t' > 1/2$. Here, equation (\ref{ineq:defn-t-prime}) and the assumed equality $\int_{1/2}^1 f_j = \int_{1/2}^1 f_i$ give us $\int_{1/2}^{t'} f_j < \int_{1/2}^{t'} f_i$. 

Note that (due to normalization) we have $\int_0^{1/2} f_j = \int_0^{1/2} f_i$. This equality contradicts an application of Lemma~\ref{lemma:R2R}, property (i), to the intervals $[0,1/2]$ and $[1/2, t']$
\begin{align*}
\frac{\int_0^{1/2} f_j}{\int_0^{1/2} f_i} & \leq  \frac{\int_{1/2}^{t'} f_j}{\int_{1/2}^{t'} f_i} < 1
\end{align*}
Therefore, by sorting the values of the eval queries $\{\Eval_i(1/2, 1)\}_{i\in [n]}$, we can efficiently find the MLRP order. 
\end{proof}
	
\subsection{Lipschitzness of Cut and Eval Queries under MLRP}
\label{appendix:mlrp-lipschitz}
Here, we prove that if the value densities satisfy MLRP, then the cut and eval queries of the corresponding cake-division instance are Lipschitz continuous. 
Recall, that the value densities $f_i$s are assumed to be Riemann integrable and hence, by definition, $f_i$s are bounded. Furthermore, as mentioned previously, MLRP mandates that the value densities are strictly positive over the cake. Therefore, for each agent $i \in [n]$ and $x \in [0,1]$ we have $0 < L \leq f_i(x) \leq U$ for some $L, U \in \mathbb{R}_+$. Proposition~\ref{Lip} below shows that we can express the Lipschitz constant $\lambda$ in terms of these bounding parameters, $\lambda \leq \max\{1/L, U, U/L\}$. This establishes the Lipschitz continuity of the cut and eval queries under MLRP.

	
			
 \begin{proposition} \label{Lip}
 	Let $\Cake = \langle [n], \{f_i \}_{i \in [n] } \rangle$ be a cake-division instance in which the value densities are positively-bounded, i.e., there exists $L,U \in \mathbb{R}_+$ such that $0 < L\leq f_i(x) \leq U$ for each agent $i \in [n]$ and all $x \in [0,1]$. Then the cut and eval queries, $\{\Cut_i\}_i$ and $\{\Eval_i\}_i$, are $\lambda$-Lipschitz with $\lambda \leq \max\{1/L, U, U/L\}$. 	
 \end{proposition}
\begin{proof}
Fix an agent $i \in [n]$. First, we will prove that the eval query (function) $\Eval_i$ is $U$-Lipschitz. Then, we will prove that $\Cut_i$ is $\max\{1/L, U/L\}$-Lipschitz. These two results imply the stated claim.

Recall that $\Eval_i: [0,1] \times [0,1] \mapsto \mathbb{R}_+$ is said to be $U$-Lipschitz iff $|\Eval_i(\ell, r) - \Eval_i(\ell', r')| \leq U \max\{ |\ell - \ell'|, |r-r'|\}$. We will establish this inequality component wise. 

For a fixed first component $\ell_0 \in [0,1]$,
 consider the function $g(r) \coloneqq \Eval_i(\ell_0, r) = \int_{\ell_0}^r f_i(x) dx$ for $r \in [0,1]$. For $r, r' \in [0,1]$, the definition of $g$ allows us to write $|g(r) - g(r')| = |\int_r^{r'} f_i(x) dx|$. Since, $f_i(x) \leq U$ for all $x \in [0,1]$, we obtain the inequality $|g(r) - g(r')| \leq U |r-r'|$. Therefore, the function $g(r)$ and hence $\Eval_i(\ell,r)$ is $U$-Lipschitz in the second argument $r$.

Similarly, considering the function $\widehat{g}(\ell) \coloneqq \Eval_i(\ell, r_0)$, with a fixed $r_0 \in [0,1]$, we obtain the desired inequality, $|\widehat{g}(\ell) - \widehat{g}(\ell')| =  |\int_{\ell}^{\ell'} f_i(x) dx| \leq U |\ell-\ell'|$ for $\ell, \ell' \in [0,1]$. Hence, $\Eval_i(\ell, r)$ is $U$-Lipschitz in the first argument $\ell$ as well. The two parts imply that $\Eval_i$ is $U$-Lipschitz. 

For $\Cut_i(\ell, \tau)$ we again establish Lipschitz continuity by considering the two arguments $\ell$ and $\tau$ separately. Fix $\ell_0 \in [0,1]$ and for target values $\tau \leq \tau'$, let $y \coloneqq \Cut_i( \ell_0, \tau) $ and  $y' \coloneqq \Cut_i( \ell_0, \tau')$. Note that $y \leq y'$ and $\tau' - \tau$ is equal to agent $i$'s value for the interval $[y, y']$. This value can be lower bounded by observing that $f_i(x) \geq L$, for all $x \in [0,1]$; specifically, agent $i$'s value for the interval $[y,y']$ is at least $L (y' - y)$. Therefore, we get that $\Cut_i( \ell_0, \tau') - \Cut_i( \ell_0, \tau) = y' - y \leq \frac{1}{L} (\tau' - \tau)$. That is, $\Cut_i(\ell, \tau)$ is $\frac{1}{L}$-Lipschitz in $\tau \in \mathbb{R}_+$.
	
For the remaining analysis, fix target value $\tau_0 \in \mathbb{R}_+$, and for $\ell \leq \ell'$, let $y \coloneqq \Cut_i( \ell, \tau_0) $ and  $y' \coloneqq \Cut_i( \ell', \tau_0)$. A useful observation here is that agent $i$'s value for the interval $[y, y']$ is equal to $i$'s value for the interval $[\ell, \ell']$, i.e., $v_i(y, y') = v_i(\ell, \ell')$. 

Specifically, if $y \leq \ell'$, then we can write $\delta$ to denote agent $i$'s value for the interval $[y, \ell']$ (i.e., $\delta \coloneqq v_i(y, \ell')$) and note that $v_i(\ell, \ell') = v_i(\ell, y) + v_i(y, \ell') = \tau_0 + \delta$. The equality then follows from observing that $v_i(y, y') = v_i(y, \ell') + v_i(\ell',y') = \delta + \tau_0$. Complementarily, if $y > \ell'$, then we write $\delta \coloneqq v_i(\ell', y)$ and use the following equations to obtain the desired equality: $v_i(\ell, \ell') = v_i(\ell, y) - v_i(\ell', y) = \tau_0 - \delta$ and $v_i(y, y') = v_i(\ell', y') - v_i(\ell', y) = \tau_0 - \delta$. 

The bounds on the value densities, $0< L \leq f_i(x) \leq U$, and the equality $v_i(y, y') = v_i(\ell, \ell')$ imply $L (y' - y) \leq v_i(y, y') = v_i(\ell, \ell')   \leq U (\ell' - \ell)$. Hence, we have $ | \Cut_i( \ell', \tau_0) - \Cut_i( \ell, \tau_0)|  = |y' - y| \leq \frac{U}{L} |\ell' - \ell|$. That is, $\Cut_i(\ell, \tau)$ is $\frac{U}{L}$-Lipschitz in $\ell \in \mathbb{R}_+$. 

Combing the above-mentioned bounds we get that the $\Cut_i$ and $\Eval_i$ queries are $\lambda$-Lipschitz with $\lambda \leq \max \{U, 1/L, U/L\}$.
\end{proof}


\section{Missing Proofs from Section~\ref{EFdivisions}}

\label{appendix:exist-compute-rd}

\subsection{Proof of Proposition~\ref{RDn}}
Here we restate and prove Proposition~\ref{RDn}
 
\PropositionCcnLipschitz*
  
\begin{proof}
Applying strong induction over $i \in \{2, 3, \ldots, n\}$, we will prove that the the function $\RD_i$ is $\lambda^{2(i-1)}$-Lipschitz. For the base case of $\RD_2$ consider the following bound, with points $x, x' \in [0,1]$:
  	\begin{align*}
  	\big| \RD_2(x') - \RD_2(x)\big| &= \big| \Cut_1\left( x', \Eval_1(0,x')  \right) -  \Cut_1\left( x, \Eval_1(0,x)  \right)\big|  \tag{by definition of $\RD_2$}\\
  	& \leq \lambda \max\{|x'-x| , |\Eval_1(0,x')-\Eval_1(0,x)|\}  \tag{$\Cut_1$ is $\lambda$-Lipstchiz}\\
  	& \leq \lambda \max\{|x'-x| , \lambda |x'-x|\}  \tag{$\Eval_1$ is $\lambda$-Lipstchiz}\\
  	& = \lambda^2 |x'-x|   \tag{since $\lambda \geq 1$} 
  	\end{align*}
Hence, $\RD_2$ is $\lambda^2$-Lipschitz. Next, assuming that, for all $k \leq i-1$, $\RD_k$ is $\lambda^{2(k-1)}$-Lipschitz, we will prove that $\RD_i$ is $\lambda^{2(i-1)}$-Lipschitz. Note that the following bound holds for all $x', x \in [0,1]$
  	\begin{align} 
  	& \Eval_{i-1} \left(\RD_{i-2}(x'), \RD_{i-1}(x')\right) - \Eval_{i-1} \left( \RD_{i-2}(x),\RD_{i-1}(x) \right) \nonumber \\
  	& \leq \lambda \max \{ |\RD_{i-2}(x') - \RD_{i-2}(x)| , |\RD_{i-1}(x') - \RD_{i-1}(x)| \} \tag{since $\Eval_{i-1}$ is $\lambda$-Lipschitz}  \nonumber\\
  	&  \leq \lambda \max \{ \lambda^{2(i-3)} |x'-x| , \lambda^{2(i-2)} |x'-x| \} \nonumber
  	\tag{using the induction hypothesis for $i-2$ and $i-1$}  \\
  	& = \lambda  \  \lambda^{2(i-2)} |x-x'| \tag{since $\lambda \geq 1$} \nonumber \\
  	& = \lambda^{2i-3} |x'-x| \label{eval}
  	\end{align}
We can now bound $\big| \RD_i(x') - \RD_i(x)\big|$ for all $x', x \in [0,1]$:
  	\begin{align*}
  	& \big| \RD_i(x') - \RD_i(x)\big| \\
  	&\ = \big| \Cut_{i-1}\left(\RD_{i-1}(x'),  \Eval_{i-1}\left(\RD_{i-2}(x'),\RD_{i-1}(x')\right)  \right) - 
  	\Cut_{i-1} \left(\RD_{i-1}(x), \Eval_{i-1}\left(\RD_{i-2}(x),\RD_{i-1}(x)\right)  \right) \big|\\
  	& \ \leq \lambda \max \{|\RD_{i-1}(x')-\RD_{i-1}(x)| , \lambda^{2i-3} |x'-x|\}   \tag{since $\Cut_{i-1}$ is $\lambda$-Lipstchiz and by equation~(\ref{eval})}\\
  	& \ \leq \lambda \max\{\lambda^{2(i-2)}|x'-x| , \lambda^{2i-3} |x'-x|\}   \tag{using the induction hypothesis for $i-1$} \\
  	& \ = \lambda^{2(i-1)} |x'-x| \tag{since $\lambda \geq 1$} 
  	\end{align*}
 Therefore, $\RD_i$ is $\lambda^{2(i-1)}$-Lipschitz for all $2 \leq i \leq n$. Setting $i =n$, gives us the desired claim. 
\end{proof}
  
 %

\subsection{Proof of Lemma~\ref{RDcomputation}}
\label{section:proofRDc}

This section restates Lemma~\ref{RDcomputation} and shows that \textsc{BinSearch} (Algorithm~\ref{alg:BinSearch}) satisfies this claim. 

\RDcomputation*

\begin{proof}
Lemma~\ref{lemma:delta-rd-exist} implies that the problem of computing a $\delta$-ripple division reduces to finding a point $x \in (0,1)$ that satisfies $\RD_n(x) \in [1- \delta, 1)$. Here, we will show that \textsc{BinSearch} (Algorithm~\ref{alg:BinSearch}) finds such a point $x$ (and, hence, a $\delta$-ripple division) in $\mathcal{O}\left( {\rm poly} ( n, \log \lambda, \log \frac{1}{\delta}) \right)$ time. 

By design, \textsc{BinSearch} maintains two points $\ell \leq r$ and keeps iterating till it finds a point $x = (\ell + r)/2$ that satisfies the required property $\RD_n\left(  x \right)  \in [1- \delta, 1)$. Hence, at termination the algorithm indeed finds a $\delta$-ripple division. Below we complete the proof by showing that the time complexity of the algorithm is as stated (in particular, the algorithm necessarily terminates). 

Note that, throughout the algorithm's execution, the maintained points $\ell$ and $r$ satisfy $\RD_n(\ell) < 1 - \delta$ and $\RD_n(r) =1$. This invariant holds at the beginning of the algorithm, where we initialize $\ell = 0$ and $r=1$; recall that $\RD_n(0) = 0$ and $\RD_n(1) = 1$. Furthermore, in each iteration of the while-loop in \textsc{BinSearch}, the left endpoint $\ell$ gets updated ($\ell \leftarrow (\ell+r)/2$) iff the midpoint $(\ell+r)/2$ satisfies  $\RD_n\left(  \frac{\ell+r}{2} \right) < 1 - \delta$. That is, even after the update we have $\RD_n(\ell) < 1 - \delta$. Similarly, for the right endpoint $r$, during the algorithm's execution we have $\RD_n(r) = 1$.  

This invariant implies that between $\ell$ and $r$ we always have a point $x \in (\ell, r)$ which satisfies $\RD_n(x) = 1 - \delta/2$. This observation is obtained by applying the intermediate value theorem to the (continuous) function $\RD_n$ on the interval $[\ell, r]$. 
Hence, the interval under consideration, $[\ell, r]$, continues to contain a required point. 

Furthermore, in each iteration of the while-loop of \textsc{BinSearch}, the difference between $\ell$ and $r$ reduces by a multiplicative factor of two. Hence, after $T$ iterations we must have $r - \ell \leq \frac{1}{2^T}$. In particular, the difference between $(\ell+r)/2$ and the desired point $x\in (\ell, r)$  is at most $\frac{1}{2^T}$, after $T$ iterations. Therefore, the following bound holds after $T$ iterations 
\begin{align*}
\left|\RD_n(x) - \RD_n \left( \frac{\ell+r}{2} \right) \right| & \leq \lambda^{2(n-1)} \left| x - \left(\frac{\ell+r}{2} \right) \right| \tag{Proposition~\ref{RDn}: $\RD_n$ is $\lambda^{2(n-1)}$-Lipschitz} \\
& \leq \frac{\lambda^{2(n-1)}}{2^T}
\end{align*}
Note that, if the iteration count $T$ is greater than $2\left(n - 1\right) \ \log \left( \frac{2 \lambda}{\delta} \right)$, then $\left|\RD_n(x) - \RD_n \left( \frac{\ell+r}{2} \right) \right|  < \delta/2$. Since $\RD_n(x) = 1 - \delta/2$, we get that $\RD_n \left( \frac{\ell+r}{2} \right) \in (1- \delta, 1)$. That is, in at most $2\left(n - 1\right) \ \log \left( \frac{2 \lambda}{\delta} \right)$ iteration \textsc{BinSearch} will certainly find a point $\left( \frac{\ell+r}{2} \right)$ that satisfies the termination criterion $\RD_n \left( \frac{\ell+r}{2} \right) \in [1-\delta, 1)$. This shows the time complexity of \textsc{BinSearch} is $\mathcal{O}\left( {\rm poly} ( n, \log \lambda, \log \frac{1}{\delta}) \right)$ and completes the proof. 
\end{proof}

\section{Distribution Families with MLRP} 
\label{appendix:mlrp-use-cases}

This section highlights that various well-studied distribution families bear MLRP. Recall, that two probability density functions $f_i$ and $f_j$ are said to satisfy MLRP iff, for every $x \leq y$ in the domain, we have 
\begin{align*}
\frac{f_j(x)}{f_i(x)} & \leq \frac{f_j(y)}{f_i(y)}.
\end{align*}
That is, the {likelihood ratio} $\nicefrac{f_j(x)}{f_i(x)}$ is non-decreasing in the argument $x \in \mathbb{R}$. It is relevant to note that this property can be verified analytically. We support this observation through two examples: Binomial polynomials (and, hence, linear value densities) in Proposition~\ref{prop:binomial-polynomials} and Gaussian densities in Proposition~\ref{prop:gaussian}. Similar analytic arguments can be used to show that many other classes of value densities satisfy MLRP \cite{larsen2001introduction, casella2002statistical,saumard2014log}. \\

 \noindent
 \textbf{Binomial Polynomials:}
The following proposition shows that every pair of binomial polynomials bear MLRP over $[0,1]$, i.e., we have a total order over this family of density functions with respect to MLRP.\footnote{As mentioned previously, MLRP is preserved under scaling and, hence, here we do not have to explicitly enforce normalization.} Hence, the results developed in the work hold for cake-division instances in which the value densities are binomial.  Also, setting the exponent parameters $s=1$ and $t = 0$ in this proposition, we observe that linear functions form a special case of binomial polynomials.
 
 
\begin{proposition} \label{prop:binomial-polynomials}
With integer exponents $s > t$, let $f_i (x) = a_i x^s + b_i x^t$ and $f_j (x) = a_j x^s + b_j x^t$ be two binomial polynomials. Then, $f_i$ and $f_j$ bear MLRP iff $a_i b_j - a_j b_i \leq 0$. 
 \end{proposition}
 \begin{proof}
 For any two points $x, y \in [0,1]$, such that $x \leq y$, the MLRP condition for binomials corresponds to the following inequality 
 \begin{align*}
 \frac{a_j x^s + b_j x^t}{a_i x^s + b_i x^t} \leq \frac{a_j y^s + b_j y^t}{a_i y^s + b_i y^t} 
 \end{align*} 
 
Since $f_i$ and $f_j$ constitute value densities with full support, the values of these functions are positive over the cake. Hence, the previous equation can be rewritten as
\begin{align*}
 (a_j x^s + b_j x^t)(a_i y^s + b_i y^t) \leq (a_j y^s + b_j y^t)(a_i x^s + b_i x^t)
 \end{align*}
This equation further simplifies to  
 \begin{align*}
 (a_i b_j - a_j b_i) x^t y^s \leq (a_i b_j - a_j b_i) x^s y^t
 \end{align*}
Finally, we rewrite the last inequality to obtain that for binomials MLRP is equivalent to  
\begin{align} \label{equation:binomial-mlrp}
 x^ty^t (a_i b_j - a_j b_i) (y^{s-t} - x^{s-t}) \leq 0
\end{align}
Since $y-x \geq 0$ and $(s-t) \in \mathbb{Z}_+$, we have that\footnote{For $(s-t) \in \mathbb{N}$, write $y^{s-t} - x^{s-t} = (y-x)(y^{s-t-1} + y^{s-t-2}x+ \dots + y x^{s-t-2} + x^{s-t-1})$. Hence, $(y-x)$ and $(y^{s-t} - x^{s-t})$ have the same signs, for $0 \leq x \leq y \leq 1$.} $y^{s-t}-x^{s-t} \geq 0 $. Furthermore, since $x^ty^t \geq 0$, inequality (\ref{equation:binomial-mlrp}) holds iff $ a_i b_j - a_j b_i \leq 0$. That is, $f_i$ and $f_j$ bear MLRP iff $ a_i b_j - a_j b_i \leq 0$. 
\end{proof}

\noindent
\textbf{Gaussian distributions:} 
The next proposition shows that Gaussian distributions with different means, but the same variance, bear MLRP. 

\begin{proposition} \label{prop:gaussian}
	Let $f_i$ and $f_j$ be two Gaussian density functions with the same variance $\sigma^2$ and means $\mu_i \leq \mu_j$, respectively. Then, $f_i$ and $f_j$ satisfy MLRP. 
\end{proposition}
 
 \begin{proof}
Here, the two density functions are $f_i(x) = \frac{1}{\sigma \sqrt{2 \pi}} e^{- \frac{(x - \mu_i)^2}{2 \sigma^2}}$ and $f_j(x) = \frac{1}{\sigma \sqrt{2 \pi }} e^{- \frac{(x - \mu_j)^2}{2 \sigma^2}}$, for $x \in \mathbb{R}$. 
Note that the likelihood ratio of $f_j$ and $f_i$ is equal to   
 \begin{align*}
 \frac{e^{- (x - \mu_j)^2/2 \sigma^2}}{e^{- (x - \mu_i)^2/2 \sigma^2}} = {\rm exp} \left( \frac{1}{2 \sigma^2} \left( (x-\mu_i)^2-(x-\mu_j)^2 \right) \right)
  \end{align*}
Write $g(x) \coloneqq \frac{1}{2 \sigma^2} \left( (x-\mu_i)^2-(x-\mu_j)^2 \right)$ and note that the derivate of this function $g'(x) = \frac{(\mu_j-\mu_i)}{\sigma^2} \geq 0$ if and only if $\mu_i \leq \mu_j$. Hence, $g$ is an increasing function in $x$---and so is ${\exp} \left( g(x) \right)$---iff $\mu_i \leq \mu_j$. That is, $f_i$ and $f_j$ bear MLRP iff $\mu_i \leq \mu_j$.
 \end{proof}

An analogous result holds for Gaussians with mean zero, but distinct variances.

\section{Bit Complexity of Cake Division}
\label{appendix:example-precision-loss}

This section provides a cake-division instance (with MLRP) in which the unique envy-free allocation has an irrational cut point. Notably the parameters that specify the value densities in this instance are rational. This example implies that, in general, one cannot expect an efficient algorithm that outputs an {exact} envy-free allocation. That is, when considering cake-division algorithms with bounded bit complexity, a precision loss is inevitable. 

We will also show, through the example, that to obtain a nontrivial bound on the envy (between the agents) the bit complexity of the output has to be $\Omega (\log \lambda)$; here $\lambda \geq 1$ is the Lipschitz constant of the cut and eval queries. Hence, a runtime dependence of $\log \lambda$ is unavoidable as well.

Consider cake division between three agents with identical value densities, $f: [0,1] \mapsto \mathbb{R}_+$. Such an instance trivially satisfies MLRP, since the likelihood ratio of the value densities is equal to $1$ throughout the cake. In particular, the density function $f$ is piecewise linear and is defined as follows
\begin{align*}
f(x) & \coloneqq
\begin{cases}
\frac{\lambda}{3(\lambda-1)} & \quad \text{if} \ 0 \leq x \leq \left(1 - \frac{1}{\lambda}\right) \\
\frac{2 \lambda^2}{3} x + \left( \lambda -  \frac{2\lambda^2}{3} \right) & \quad \text{if} \ \left(1 - \frac{1}{\lambda}\right) \leq x \leq 1 \\
\end{cases}
\end{align*}

Note that this function can be given as input using only rational parameters. In addition, $f$ is discontinuous at $1 - \frac{1}{\lambda}$; recall that our results require the value densities to be integrable, and not necessarily continuous.

Here, the values of the agents are normalized, $\int_{0}^{1}f(x)dx =1$. Also, the following bounds hold for the density: $0<\frac{\lambda}{3(\lambda-1)} \leq f(x) \leq \lambda$ for all $x \in [0,1]$. Therefore, Proposition~\ref{Lip} implies that the cut and eval queries in this instance are $3 \lambda$-Lipschitz.
Since in this instance the three agents have identical value densities (with full support over $[0,1]$), there exists a unique envy-free allocation wherein each agent receives an interval of value exactly equal to $1/3$. Write $0=x^*_0 \leq x^*_1 \leq x^*_2 \leq x^*_3=1$ to denote the cut points of this (unique) envy-free allocation; in particular, agent $i \in [3]$ receives the $i$th interval $[x^*_{i-1},x^*_i]$. 

First, we will show that $x^*_2$ is irrational. Note that the interval $[0, 1 - 1/\lambda]$ is of value $1/3$: $\int_0^{1-\frac{1}{\lambda}} f(x) dx = \frac{1}{3}$. Hence, the first cut point $x^*_1=1-\frac{1}{\lambda}$. The second cut point $x^*_2$ now lies in the interval ($[x^*_1, 1]$) of width $\frac{1}{\lambda}$ and it must satisfy $\int_{1-\frac{1}{\lambda}}^{x^*_2} f(x) dx =\frac{1}{3}$. Using the definition of $f$ in this range, we get that $x^*_2$ is a solution of the following quadratic equation 
\begin{align}
\lambda^2 (x^*_2)^2 + (3 \lambda - 2 \lambda^2) x^*_2 + (\lambda^2 - 3 \lambda +1) = 0  \label{equation:quad-irrational}
\end{align}
With the constraint $x^*_2 \in \left(1- \frac{1}{\lambda},1 \right)$, equation~(\ref{equation:quad-irrational}) gives us  $x^*_2 = 1 - \left( \frac{3- \sqrt{5}}{2}\right) \frac{1}{\lambda}$. That is, the cut point $x^*_2$ is irrational, for $\lambda \in \mathbb{Q}_+$.

We next establish a lower bound on the bit complexity of the output. Consider, in the above-mentioned instance, any allocation $\mathcal{I}=\{I_1, I_2, I_3\}$ wherein the envy between the agents is, say, less than $\eta = 1/2$, i.e., $v_i(I_i) \geq v_i(I_j) - 1/2$ for all $i, j \in [n]$.\footnote{Here, the choice of $\eta=1/2$ is essentially for ease of exposition; by scaling down the density in the range $\left[ 0, \left(1 - {1}/{\lambda} \right) \right]$, we can drive $\eta$ close to one.}  For any such allocation $\mathcal{I}$ one of the cut points must lie in $[\left(1 - {1}/{\lambda} \right), 1]$. Indeed, this interval is of value $2/3$ to each agent, $\int_{\left(1- 1/\lambda \right)}^1 f = 2/3$. The bit complexity of such a cut point is $\Omega(\log \lambda)$. Hence, in general, the bit complexity of the any algorithm that finds an allocation (equivalently, outputs cut points) with bounded envy is $\Omega(\log \lambda)$. \\


\noindent
{\bf Welfare-maximizing allocations induced by irrational cuts:} One can also construct cake-division instances (with rational and MLRP value densities) in which the welfare-maximizing allocations have irrational cut points.  

In particular, consider cake division between two agents with identical value densities, $f(x) = x+1/2$. Here, the (unique) allocation $\mathcal{I} =\{I_1, I_2 \}$ that maximizes egalitarian welfare consists of the intervals $I_1 = \left[0, \frac{\sqrt{5}-1}{2}\right]$ and $I_2 = \left[\frac{\sqrt{5}-1}{2}, 1\right]$. Allocation $\mathcal{I}$ is also the (unique) equitable, proportional, and envy-free allocation in this instance. 

In addition, consider a cake-division instance with two agents and the following value densities: $f_1(x)=1$ and $f_2(x)=3x^2$. Note that these densities satisfy MLRP. In this instance, the ({unique}) social welfare maximizing allocation $\mathcal{S} = \{S_1, S_2\}$ is obtained by an irrational cut; specifically, $S_1= \left[0, \frac{1}{\sqrt{3}} \right]$ and $S_2 = \left[ \frac{1}{\sqrt{3}},1 \right]$. The cut point $\frac{1}{\sqrt{3}}$ is the switching point (as defined in Section \ref{section:social welfare}) between the two densities $f_1$ and $f_2$.

\section{Robustness of MLRP}
\label{appendix:structured-perturbations-for-MLRP} 


This section shows that our framework extends to the class of non-full-support value densities considered in \cite{alijani2017envy}. In particular, Alijani et al. \cite{alijani2017envy} established that a contiguous envy-free cake division can be efficiently computed if every agent $i \in [n]$ uniformly values a single interval $[\ell_i, r_i] \subseteq [0,1]$ and these intervals satisfy the following ordering property {\rm OP}: for all $i, j \in [n]$ we have $\ell_i \leq \ell_j$ iff $r_i \leq r_j$. This ordering property {\rm OP} ensures that, in particular, any interval $[\ell_i, r_i]$ is not a strict subset of $[\ell_j,r_j]$ for $i, j \in [n]$.

Let $\mathcal{K} = \langle [n], \{f_i \}_{i \in [n] } \rangle $ denote a cake-division instance with such value densities and note that, for each $i \in [n]$,\begin{equation} \label{equation:pconstant}
f_i(x) \coloneqq
\begin{cases}
\frac{1}{r_i-{\ell}_i} &\quad \text{ if } \ x \in [\ell_i, r_i] \\
0 &  \quad  \text{otherwise}\\
\end{cases}
\end{equation}
Write $h_i \coloneqq 1/(r_i - \ell_i)$. The work of Alijani et al. \cite{alijani2017envy} shows that an envy-free allocation of $\mathcal{K}$ can be computed efficiently.

Indeed, the value densities in $\mathcal{K}$ do not have full support over the cake. However, we will show that we can perturb these densities $f_i$s  (in a structured manner) to obtain an instance $\widehat{\mathcal{K}} = \langle [n], \{\widehat{f}_i \}_{i \in [n] } \rangle $ such that (i) the value densities $\widehat{f}_i$s in $\widehat{\mathcal{K}}$ have full support and bear MLRP (Claim~\ref{claim:perturbation-mlrp}) and (ii) any envy-free allocation in $\widehat{\mathcal{K}}$ forms an envy-free allocation, up to a small precision loss, in the original instance $\mathcal{K}$ (Claim~\ref{claim:perturbation-envy-bound}).



Index the agents such that $0 \leq \ell_1 \leq \ell_2\leq \ldots \leq \ell_n \leq 1$. Property {\rm OP} ensures that, under this indexing, the right endpoints $r_i$s are also sorted: $r_1 \leq r_2 \ldots \leq r_n$. To construct the densities $\widehat{f}_i$, for each agent $i$, we will define counting functions $c_i: [0, 1] \mapsto \{0, 1, \ldots, i\}$ and $d_i: [0, 1] \mapsto \{0, 1, 2, \ldots, n-i\}$ as follows: $c_i(x) \coloneqq \sum_{k=1}^i \mathbbm{1} \{ \ell_k \leq x \}$ and $d_i(x) \coloneqq \sum_{k=i}^n \mathbbm{1} \{ r_k \leq x \}$, i.e., $c_i(x)$ is the number of left endpoints in $\{\ell_1, \ell_2, \ldots, \ell_i \}$ that appear before $x$ and $d_i(x)$ is equal to the number of right endpoints in $\{r_i, r_{i+1}, \ldots, r_n \}$ that appear before $x$.

For a sufficiently large  parameter $H \geq 1$ (which we will set in Claim~\ref{claim:perturbation-envy-bound}), we define $\widehat{f}_i: [0,1] \mapsto \mathbb{R}_+$ as follows\footnote{The integrals of the value densities $\widehat{f}_i$s are not normalized to $1$, though, as mentioned previously, MLRP is preserved under scaling and, hence, we do not have to explicitly enforce normalization.} 
\begin{equation} \label{equation:perturbation}
\widehat{f}_i(x) \coloneqq  
 h_i \ H^{(c_i(x) - i - d_i(x)) }  
\end{equation}

In other words, we initialize $\widehat{f}_i(0) = {h_i} {H}^{-i}$ and scale up this density by a multiplicative factor of ${H}$ whenever an interval $[\ell_k, r_k]$ starts, for any agent $k \leq i$. Note that $c_i(x) = i$ and $d_i(x) = 0$ for all $x \in [\ell_i, r_i]$ and, hence, in this interval $\widehat{f}_i(x) = h_i = f_i(x)$. Beyond this range, i.e., for $x \geq r_i \geq \ell_i$, we have $c_i(x) = i$. Therefore, once the interval $[\ell_i, r_i]$ ends, we scale down $\widehat{f}_i$ by a multiplicative factor of $H$, whenever an interval $[\ell_j,r_j]$ ends, for any agent $j \geq i$. Furthermore, for all $x \notin [\ell_i, r_i]$, we have $\widehat{f}_i(x) \leq h_i/H$. Also, as required, $\widehat{f}_i$s have full support over $[0,1]$.

Applying Proposition~\ref{Lip}, we obtain that the Lipschitz constant of cut and eval queries in $\widehat{\mathcal{K}}$ is $\lambda = H^n \ \max_{1 \leq i \leq n} \ h_i$. The next claim shows that the the value densities in $\widehat{\mathcal{K}}$ satisfy MLRP.

\begin{claim} \label{claim:perturbation-mlrp}
	Let $\mathcal{K} = \langle [n], \{f_i \}_{i \in [n] } \rangle$ be a cake-division instance in which the value densities satisfy equation~(\ref{equation:pconstant}) and the ordering property {\rm OP}. Then, for parameter ${H} \geq 1$, the value densities $\widehat{f}_i$s (as defined in equation~(\ref{equation:perturbation})) satisfy MLRP. 
\end{claim}

\begin{proof}
In this constructed instance $\widehat{\mathcal{K}} = \langle [n], \{\widehat{f}_i \}_{i \in [n] } \rangle$ the agents are indexed such that $0 \leq \ell_1 \leq \ell_2\leq \dots \leq \ell_n \leq 1$ and, by {\rm OP}, we have $r_1 \leq r_2 \leq \ldots \leq r_n$. Fix any two agents $i, j \in [n]$ such that $i<j$. The indexing among the agents ensures that $\ell_i \leq \ell_j$ and $r_i \leq r_j$. 

To prove the stated claim it suffices to show that the likelihood ratio $\frac{\widehat{f}_j(x)}{\widehat{f}_i(x)}$ is nondecreasing throughout $[0,1]$. We establish the monotonicity of this likelihood ratio by considering three different ranges in the cake $[0, \ell_i]$, $[\ell_i, r_j]$, and $[r_j, 1]$.

Initially, the likelihood ratio $\frac{\widehat{f}_j(0)}{\widehat{f}_i(0)} = \left(\frac{h_j}{h_i}\right) \frac{1}{{H}^{j-i}}$. For all $x \in [0, \ell_i]$ whenever the density $\widehat{f}_i(x)$ experiences a multiplicative increase so does $\widehat{f}_j(x)$, by the same factor $H$. In particular, if $c_i$ increases then so does $c_j$ throughout the interval $[0, \ell_i]$: if at a point $x \in [0, \ell_i]$ the density $\widehat{f}_i$ is scaled up by a multiplicative factor ${H}$ (i.e., $c_i$ increases by one), then it must have been the case that an interval $[\ell_k, r_k]$, with $k \leq i$, starts at $\ell_k = x$. In such a case, by definition, $\widehat{f}_j$ also increases by a multiplicative factor ${H}$ (i.e., $c_j$ increases by one as well). 
Therefore, for all $x \in  [0, \ell_i]$, the likelihood ratio stays constant at $\frac{\widehat{f}_j(0)}{\widehat{f}_i(0)}$.

In the interval $[\ell_i, r_j]$, the density $f_i$ does not increase, since here $c_i$ saturates at $i$ after $\ell_i$ and $d_i$ increases from zero beyond $r_i$. At the same time, in this interval $[\ell_i, r_j]$ the density $f_j$ is nondecreasing. Therefore, the likelihood ratio continues to be monotonic in the interval $[\ell_i, r_j]$ as well. 

Finally, for $x \geq r_j$, we note that $d_i$ and $d_j$ increase synchronously. Also, for points $x \in [r_j, 1]$ we have $c_j(x) = j$ and $c_i(x) = i$. Therefore, in this range the likelihood ratio stays constant at $\frac{\widehat{f}_j(r_j)}{\widehat{f}_i(r_j)}$. Overall, we obtain the monotonicity of the likelihood ratio and the MLRP guarantee follows. 
\end{proof}

 \begin{claim} \label{claim:perturbation-envy-bound}
Given a cake-division instance $\mathcal{K} = \langle [n], \{f_i \}_{i \in [n] } \rangle$ in which the value densities satisfy equation~(\ref{equation:pconstant}) and the ordering property {\rm OP}. Let $\widehat{\mathcal{K}} = \langle [n], \{\widehat{f}_i \}_{i \in [n] } \rangle$ be the cake-division instance defined above, with parameter $H \coloneqq  \frac{2}{\eta} \ \max_{1 \leq i \leq n} h_i$, and suppose that allocation $\mathcal{I}=\{I_1, \ldots, I_n\}$ is envy-free up to an additive factor of $\eta$ in $\widehat{\mathcal{K}}$ (i.e., $v_i(I_i) \geq v_i(I_j) - \eta$ for all $i, j \in [n]$). Then, allocation $\mathcal{I}$ envy-free up to $2\eta$ in $\mathcal{K}$. 
\end{claim}

\begin{proof}
For agent $i \in [n]$, write $v_i$ and $\widehat{v}_i$ to denote the valuation function of agent $i$ under $f_i$ and $\widehat{f}_i$, respectively. By construction, the densities $f_i$ and $\widehat{f}_i$ coincide over the interval $[\ell_i ,r_i]$; in particular, $f_i(x) = \widehat{f}_i(x)=h_i$ for all $x \in [\ell_i ,r_i]$. Hence, for any interval $I$ we have $v_i(I \cap [\ell_i, r_i]) = \widehat{v}_i(I \cap [\ell_i, r_i])$. 

Furthermore, the construction (equation~(\ref{equation:perturbation})) of $\widehat{f}_i$ ensures that $\widehat{v}_i(I \setminus [\ell_i,r_i]) \leq \frac{h_i}{{H}} \leq \eta/2$. The last inequality follows from the choice of the parameter $H$. Note that $f_i(x) = 0$ for all $x \notin [\ell_i, r_i]$ and, hence, $v_i(I \setminus [\ell_i,r_i]) =0$. These bounds imply that, for any agent $i \in [n]$, the difference in values of any interval $I$ is at most $\eta/2$:
\begin{align} \label{equation:perturb-value}
 |\widehat{v}_i(I) - v_i(I)| & \leq \frac{\eta}{2} 
\end{align}
 
	

Hence, an allocation $\mathcal{I}=\{I_1, \ldots, I_n\}$ that is envy-free up to an additive factor of $\eta$ in $\widehat{\mathcal{K}}$ (i.e., $\widehat{v}_i(I_i) \geq \widehat{v}_i(I_j) - \eta$ for all $i, j \in [n]$) is envy-free up to $2\eta$ in $\mathcal{K}$. Specifically, for any $i, j \in [n]$,
\begin{align*} 
v_i({I}_i) & \geq \widehat{v}_i({I}_i) - \eta/2 \tag{using inequality~(\ref{equation:perturb-value}) with interval $I_i$} \\
& \geq  \widehat{v}_i({I}_j)- \frac{3\eta}{2} \tag{by $\eta$-envy-freeness of $\mathcal{I}$} \\
&\geq v_i({I}_j) - 2\eta \tag{using inequality (\ref{equation:perturb-value}) with interval ${I}_j$} 
\end{align*}
    That is, ${\mathcal{I}}$ is envy-free, up to $2\eta$ precision, in $\mathcal{K}$. This completes the proof.
\end{proof}

Since the constructed instance $\widehat{\mathcal{K}}$ satisfies MLRP, we can use Algorithm~\ref{alg:BinSearch} (Section~\ref{EFdivisions}) to efficiently compute, up to an arbitrary precision, an envy-free allocation $\mathcal{I}$ in the instance $\widehat{\mathcal{K}}$. The previous claim ensures that $\mathcal{I}$ is an envy-free allocation (up to an arbitrary precision) in the original instance $\mathcal{K}$ as well.  This observation highlights the fact that the ideas developed in this work are somewhat robust and extend to other value-density settings.

\section{Nonexistence of Contiguous Perfect Divisions under MLRP}
\label{appendix:perfect-cuts-nonexample}

A cake division  $\mathcal{D} = \{D_1, D_2, \dots, D_n\}$ (consisting of connected or disconnected pieces) is said to be \emph{perfect} if all the agents agree on the value of every piece, i.e., $v_i(D_j)  = 1/n$ for all $i, j \in [n]$. 

Perfect divisions are not guaranteed to exist under the contiguity requirement, i.e., there are cake-division instances that do not admit perfect allocations. In this section, we will show that the nonexistence continues to hold even with MLRP. That is, we will provide a cake-division instance with MLRP that does not admit a perfect allocation.

Consider a cake-division instance $\Cake=\langle [2], \{f_1, f_2\} \rangle$ with two agents. For a fixed constant $\alpha \in (0,1)$, we define the value densities of the two agents $f_1, f_2: [0,1] \mapsto \mathbb{R}_+ $ as follows
\begin{equation*}
f_1(x) = \begin{cases}
1+ \alpha  & \text{if} \ 0 \leq x \leq  (1- \alpha) \\
\alpha & \text{if} \ (1-\alpha) < x \leq 1
\end{cases} \qquad \text{ and } \qquad \ \
f_2(x) = \begin{cases}
1- \alpha  & \text{if} \ 0 \leq x \leq  (1- \alpha) \\
2- \alpha & \text{if} \ (1-\alpha) < x \leq 1
\end{cases}
\end{equation*}

The agents' values for the cake $[0,1]$ are normalized, $\int_0^1 f_1(x)dx = \int_0^1 f_2(x)dx = 1$. Also, the likelihood ratio of $f_1$ and $f_2$ satisfies 
\begin{equation*}
\frac{f_2(x)}{f_1(x)} = \begin{cases}
\frac{1- \alpha}{1+\alpha}  & \text{if} \ 0 \leq x \leq  (1- \alpha) \\
\frac{2- \alpha}{\alpha}& \text{if} \ (1-\alpha) < x \leq 1 
\end{cases}
\end{equation*}
Since $\frac{2- \alpha}{\alpha} > \frac{1- \alpha}{1+\alpha} $, for all $\alpha \in (0,1)$, the likelihood ratio is nondecreasing over $[0,1]$ and, hence, the densities bear MLRP. 

We assume, towards a contradiction that there exists a perfect allocation in $\Cake$. That is, there exists a point $x \in [0,1]$ such that $v_1(0,x) = v_2(0,x) = 1/2$, and $v_1(x,1) = v_2(x,1) =1/2$. The fact that the value of intervals $[0,x]$ and $[x,1]$ is equal to $1/2$ ensures that the point $x$ cannot be $0$ or $1$. We consider two complementary and exhaustive cases: Case (i) $0<x \leq (1- \alpha)$ and Case (ii) $ (1- \alpha) < x \leq 1$. The analysis is similar in both the cases. Hence, we only address Case (i) and omit the analysis for Case (ii).


For the first interval $[0,x]$ (with $0< x \leq (1- \alpha)$), we have $v_1(0,x) = x(1+\alpha)$ and $v_2(0,x) = x(1- \alpha)$. However, for any $\alpha \in (0,1)$ and $x > 0$, the following strict inequality holds: $v_1(0,x) = x(1+\alpha) > x(1- \alpha) = v_2(0,x)$. This leads to a contradiction and proves that $\Cake$ does not admit a perfect division with connected pieces.

Even though we might not have a perfect allocation, the work of Alon~\cite{alon1987splitting} proves that a perfect division with $n(n-1)$ cuts always exists. Hence, in the above-mentioned instance $\Cake$ with $2$ agents, $2$ cuts should suffice to form a perfect division. In particular, we note that the following division $\mathcal{D}^* = \{D^*_1, D^*_2\}$ is perfect in $\Cake$; here $D^*_1 = \left[ \frac{1}{2}- \frac{\alpha}{2}, 1-\frac{\alpha}{2} \right]$ and $D^*_2 = \left[ 0,\frac{1}{2}- \frac{\alpha}{2} \right] \cup \left[ 1-\frac{\alpha}{2},1 \right]$. Here, 
\begin{align*}
v_1(D^*_2) &= (1+\alpha)\left(\frac{1}{2}- \frac{\alpha}{2}\right) + \alpha\left(\frac{\alpha}{2}\right) = \frac{1}{2}
\end{align*}
and
\begin{align*}
v_2(D^*_2) &= (1-\alpha)\left(\frac{1}{2}- \frac{\alpha}{2}\right) + (2-\alpha)\left(\frac{\alpha}{2}\right) = \frac{1}{2}
\end{align*}
That is, both the agents value the piece $D^*_2$ at $1/2$. Since the valuations are normalized, we additionally have $v_1(D^*_1)=v_2(D^*_1)=1/2$. This shows that $\mathcal{D}^*$ is a perfect division (with disconnected pieces) in $\Cake$.

 



\end{document}